%% file: main.tex
\documentclass{article}
\usepackage[utf8]{inputenc}

\usepackage{mystyle}

\title{Improved Analysis of RANKING for Online Vertex-Weighted Bipartite Matching}
\author{Billy Jin\thanks{Supported in part by an NSERC fellowship PGSD3-532673-2019 and NSF grant CCF-1908517. Address: School of Operations Research and Information Engineering, Cornell University, Ithaca, NY, 14853, USA. Email: {\tt bzj3@cornell.edu}.}\\Cornell University \and David P.\ Williamson\thanks{Supported in part by NSF grant CCF-1908517. Address: School of Operations Research and Information Engineering, Cornell University, Ithaca, NY, 14853, USA.  Email: {\tt davidpwilliamson@cornell.edu}.}\\Cornell University}
\date{}

\begin{document}


\maketitle

\begin{abstract}
In this paper, we consider the online vertex-weighted bipartite matching problem in the random arrival model.  We consider the generalization of the RANKING algorithm for this problem introduced by Huang, Tang, Wu, and Zhang \cite{HTWZ19}, who show that their algorithm has a competitive ratio of 0.6534.  We show that assumptions in their analysis can be weakened, allowing us to replace their derivation of a crucial function $g$ on the unit square with a linear program that computes the values of a best possible $g$ under these assumptions on a discretized unit square.  We show that the discretization does not incur much error, and show computationally that we can obtain a competitive ratio of 0.6629. To compute the bound over our discretized unit square we use parallelization, and still needed two days of computing on a 64-core machine.  Furthermore, by modifying our linear program somewhat, we can show computationally an upper bound on our approach of 0.6688; any further progress beyond this bound will require either further weakening in the assumptions of $g$ or a stronger analysis than that of Huang et al.
\end{abstract}


\section{Introduction}
\label{sec:intro}
\input{introduction}

\section{Background}
\label{sec:background}
\input{background}

\section{Relaxing Assumptions}
\label{sec:hyp}
\input{relaxing_assumptions}

\section{LP Formulation}
\label{sec:lp}
\input{lp}

\section{Checking the Bound}
\label{sec:error}
\input{error}

\section{Computational Results}
\label{sec:computation}
\input{computation}

\section{Limits of our Method}
\label{sec:upper}
\input{upper}

\section{Conclusion}
\label{sec:conc}
\input{conclusion}

\bibliographystyle{alpha}
\bibliography{references}

\appendix

\section{Approximating the Integral of a Lipschitz function}
\label{app:lipschitz}
\input{lipschitz}

\end{document}

%% file: introduction.tex
In the maximum bipartite matching problem, we are given as input a bipartite graph $G=(U,V,E)$  such that each edge $(u,v) \in E$ has $u \in U$ and $v \in V$.  A  set $F \subseteq E$ of edges is a {\em matching} if there is at most one edge of $F$ incident to each vertex $u \in U$ and $v \in V$. The goal is to find a matching of maximum cardinality.  This problem has been well-studied and is one of the fundamental problems in combinatorial optimization (see, for example, Schrijver \cite[Chapter 16]{Schrijver03}).

In a classic paper from 1990, Karp, Vazirani, and Vazirani \cite{KVV90} introduce an online version of this problem and the RANKING algorithm for it.  In their online version of the problem, the vertices $V$ are known to the algorithm in advance, while the vertices of $U$ are introduced one at a time; we refer to the vertices of $V$ as the {\em offline} vertices and those of $U$ as the {\em online} vertices.  The algorithm maintains a matching $F$, initially empty.   As each vertex of $U$ arrives, the edges incident to $U$ are also revealed to the algorithm.  Once a vertex of $U$ arrives, the algorithm must either choose an edge incident to $U$ to add to $F$ or decide not to add an edge incident to $U$ to the matching $F$.  These choices are irrevocable: no edge incident to $U$ may be added at any later point in time.  In the RANKING algorithm, the algorithm initially chooses a random permutation $\pi$ of the offline vertices $V$; when a new vertex $u \in U$ arrives, the algorithm adds edge $(u,v)$ to 
the matching that maximizes $\pi(v)$ over the vertices $v \in V$ that do not have any edge of $F$ already incident (i.e.\ the {\em unmatched} vertices of $V$ incident to $u$), if such a vertex exists, otherwise it leaves $u$ unmatched.  Karp, Vazirani, and Vazirani prove that this algorithm achieves a {\em competitive ratio} of at least $1 - \frac{1}{e}$; that is, the algorithm finds a matching whose expected cardinality is at least $1 - \frac{1}{e}$ times the size of the maximum matching in $G$. They further show that this ratio is tight; that is, there are instances of the problem such that no online algorithm can achieve a better competitive ratio.

Since this work, there have been many simplifications of the original analysis (e.g.\ Birnbaum and Mathieu \cite{BM08}; Devanur, Jain, and Kleinberg \cite{DJK13}), proposed changes in the online model, and extensions to more general matching problems.  Of interest to us in this paper are the {\em random arrival model},
%
proposed by Goel and Mehta \cite{GM08}, and the maximum {\em vertex-weighted} online matching problem, introduced by Aggarwal, Goel, Karande, and Mehta \cite{AGKM11}.  In the random arrival model, the online vertices of $U$ arrive in an order given by a random permutation. Goel and Mehta show that the greedy algorithm attains a competitive ratio of $1-\frac1e$ in the random arrival model. Later, Karande, Mehta, and Tripathi \cite{KMT11} and Mahdian and Yan \cite{MY11} show that the RANKING algorithm has competitive ratio strictly better than $1 - \frac{1}{e}$ in this model, with Mahdian and Yan giving a competitive ratio of 0.696.  In the vertex-weighted version of the problem, the offline vertices $v \in V$ have weight $w_v \geq 0$, and the goal is to find a matching $F$ that maximizes the total weight of the matched vertices in $V$ (that is, the vertices in $V$ that have an incident edge in $F$).  Aggarwal et al.\ show that a generalization of RANKING achieves a $1 - \frac{1}{e}$ competitive ratio for the vertex-weighted version of the problem (with adversarial arrivals).  Devanur, Jain, and Kleinberg \cite{DJK13} later interpreted the Aggarwal et al.\ algorithm as follows.  Each offine vertex $v \in V$ draws a value $y_v$ from [0,1] uniformly at random; when a new vertex $u \in U$ arrives, we add edge $(u,v)$ to matching $F$ for the unmatched $v$ (if any) that maximizes $w_v(1 - g(y_v))$, where $g(y) = e^{y-1}$.  

Huang, Tang, Wu, and Zhang \cite{HTWZ19} studied the combination of these two models, the maximum vertex-weighted online matching problem in the random arrival model.  Drawing on the ideas of Devanur et al., they proposed the following further generalization of the RANKING algorithm.  In addition to having each offline vertex $v \in V$ draw a value $y_v$ from [0,1], since the online vertices arrive in random order, they propose having each online vertex $u \in U$ draw a value $y_u \in [0,1]$ uniformly at random, and have the offline vertices arrive in order of nondecreasing $y_u$.  When a new vertex $u \in U$ arrives, we add edge $(u,v)$ to the matching $F$ for the unmatched $v$ (if any) that maximizes $w_v(1 - g(y_v, y_u))$, for a function $g$ with certain properties.  Huang et al.\ assume that $g(x,y) = \frac{1}{2}(h(x) + 1 - h(y))$ for $h:[0,1] \rightarrow [0,1]$, and end up choosing $h(x) = \min(1,\frac{1}{2}e^x)$ to achieve a competitive ratio of 0.6534, beating the $1-\frac{1}{e} \approx 0.632$ competitive ratio achieved by Aggarwal et al.\ in the adversarial arrival model.

We build upon the work of Huang et al.\ to give a competitive ratio of 0.6629 for the maximum vertex-weighted online matching problem in the random arrival model.  We begin by showing that several assumptions Huang et al.\ make about the form of $g(x,y)$ needed for the analysis of their generalization of RANKING can be relaxed.  Instead, we can make several weaker assumptions about the form of $g(x,y)$.  These assumptions can be encoded in a linear program that allows us to produce the best possible piecewise-affine function $g: [0,1]^2 \rightarrow [0,1]$ under these assumptions for any given discretization of $[0,1]^2$.  

We then need to compute the competitive ratio by finding a point in $[0,1]^2$ where $g$ reaches a certain minimum of a complicated function of $g$ given by Huang et al.  To do this, we show that the error in the competitive ratio achieved by restricting ourselves to finding the minimum in the set of discretized points is linear in the size of the discretization, so we can restrict ourselves to checking just the points in this set if we are willing to tolerate some small error.  However, even checking all the discretized points becomes computationally infeasible if we discretize the square finely enough so that the error is tolerable.  
We note that the checking is easily parallelizable, and we wrote our code to use all the cores of the machine on which it is run. Even so, we still needed two days of a 64-core, 64GB machine on Amazon's EC2 platform to achieve our competitive ratio of 0.6629.

Because we use a linear program to find the function $g$, we can also use a slight modification of it to find an upper bound on the best possible competitive ratio obtainable using the Huang et al.\ analysis with our weakened assumptions on $g$.  We modify the linear program so that any function $g$ with our weakened assumptions is feasible, and modify the objective function so that it gives an upper bound on the ratio obtained via the Huang et al.\ analysis.  Solving the linear program results in an upper bound of 0.6688.  Thus any further improvement in the competitive ratio will require either further weakening in the assumptions of $g$ or a stronger analysis than that of Huang et al.

Our paper is organized as follows.  In Section \ref{sec:background}, we recap the argument of Huang et al.\ that we will use.  In Section \ref{sec:hyp}, we introduce the weaker assumptions on the function $g$ that we will use, and prove that the arguments of Huang et al.\ continue to hold under these weaker assumptions so that we can still use their bound on the competitive ratio under these weaker assumptions.  In Section \ref{sec:lp}, we introduce the LP that will define our function $g$; we show how to define a piecewise-affine function $g$ from the LP solution, and we show that the assumptions we need on $g$ hold for this LP-defined function.  In Section \ref{sec:error}, we provide a bound on the error we incur in the competitive ratio by only checking the Huang et al.\ bound at discrete points of the unit square.  In Section \ref{sec:computation}, we explain the computation that was used to obtain our competitive ratio of 0.6629. Section \ref{sec:upper} explains how we modify our linear program to obtain an upper bound on the competitive ratio that is attainable via the Huang et al.\ analysis with our weakened assumptions on $g$. We conclude in Section \ref{sec:conc}.

%% file: background.tex
As stated in the introduction, we assign each offline vertex $v \in V$ a value $y_v$ from [0,1] chosen uniformly at random, and following Huang et al.\ we assume that each online vertex $u \in U$ also has a value $y_u$ from [0,1] chosen uniformly at random, and that the online vertices arrive in nondecreasing order of their $y_u$ value.  The variant of the RANKING algorithm for the problem uses a function $g:[0,1]^2 \rightarrow [0,1]$ that is increasing in the first argument and decreasing in the second.  When an online vertex $u \in U$ arrives, it is matched to the unmatched neighbor $v \in V$ that maximizes $w_v(1-g(y_v,y_u))$.

The analysis of this algorithm by Huang et al.\ \cite{HTWZ19} follows that of Devanur, Jain, and Kleinberg \cite{DJK13}.  It considers the linear programming relaxation of the vertex-weighted bipartite matching problem and its dual linear program, shown below, with the primal on the left and the dual on the right.
\begin{align*}
\mbox{Max }  \sum_{(u,v) \in E} w_v x_{uv}  & & \mbox{Min }  \sum_{u \in U} \alpha_u + \sum_{v \in V} \alpha_v \\
\mbox{s.t. } \sum_{v: (u,v) \in E} x_{uv} \leq 1 &  \qquad\forall u \in U & \mbox{s.t. } \alpha_u + \alpha_v \geq w_v & \qquad \forall (u,v) \in E \\
\sum_{u:(u,v) \in E} x_{uv} \leq 1 &\qquad \forall v \in V & \alpha_u, \alpha_v \geq 0 & \qquad \forall u \in U, v \in V. \\
x_{uv} \geq 0 &  \qquad \forall (u,v) \in E.
\end{align*}
The goal of the analysis is to find a set of nonnegative variables $\alpha$, whose values may depend on the random $y$ values, such that $\sum_{(u,v) \in F}  w_v = \sum_{u \in V} \alpha_u + \sum_{v \in V} \alpha_v$ and $E_y[\alpha_u + \alpha_v] \geq \beta \cdot w_v$ for all $(u,v) \in E$. (Here, $F$ is the set of edges in the matching found by the algorithm.) Given the two conditions, it is possible to define a dual solution that is a factor of $\beta$ away from the total weight of the matched edges, implying a competitive ratio of $\beta$.  Whenever the algorithm adds a matching edge $(u,v)$ to $F$, it defines $\alpha_u = w_v \cdot g(y_v,y_u)$ and $\alpha_v = w_v(1 - g(y_v,y_u))$, ensuring that the first condition is met.  

The main result of Huang et al.\ is the following.
\begin{lemma}[Lemma 4.1 \cite{HTWZ19}] 
\label{lem:4.1}
Suppose that $g(x,y) = \frac{1}{2}(h(x) + 1-h(y))$, for some increasing function $h: [0,1] \to [0,1]$ that satisfies $h'(x) \leq h(x)$. Then for any $u \in U$ and $v \in V$ such that $(u,v) \in E$, $$\frac{1}{w_v}E_y[\alpha_u + \alpha_v] \geq \min_{0 \leq \gamma, \tau \leq 1} f(\gamma,\tau)$$ for $$f(\gamma,\tau) = \biggl\{(1-\tau)(1-\gamma) + (1-\tau)\int_0^\gamma g(x, \tau) dx +
\int_0^\tau \min_{\theta \leq \gamma} \left\{(1-g(\theta, y)) + \int_0^\theta g(x, y)dx + \int_\theta^\gamma g(x, \tau) dx\right\} dy \biggr\}.$$
Thus, the competitive ratio of RANKING is at least $\min_{0 \leq \gamma, \tau, \leq 1} f(\gamma, \tau)$. 
\end{lemma}  
Huang et al.\ show that by taking $h(x) = \min(1, \frac{1}{2}e^x)$, they can prove that $f(\gamma,\tau) > 1 - \frac{1}{2}\ln 2 \approx 0.6534$ for all $0 \leq \gamma, \tau \leq 1$, attaining their claimed competitive ratio by the reasoning above.

%% file: relaxing_assumptions.tex
Huang et al.\ assume that $g(x, y) = \frac12(1+h(x)-h(y))$, for some increasing function $h: [0,1] \to [0,1]$ that satisfies $h'(x) \leq h(x)$. This is a strong assumption and gives several nice properties of $g$ which are useful in the analysis. We relax this assumption and do not constrain $g$ to satisfy this condition. Instead, we replace this condition by several weaker conditions. This allows us to search over a wider class of functions $g$ when trying to maximize the bound in Lemma \ref{lem:4.1}. However, to leverage their result, we must show that the conclusion of Lemma \ref{lem:4.1}  still holds for all $g$ that satisfy these weaker conditions.   We prove the following.

\begin{theorem}
\label{thm:main_bound}
Let $g$ be a function obeying the following conditions.
\begin{enumerate}
    \item $g(x, y): [0,1]^2 \to [0,1]$ is continuous, 
    \item $g(x, y)$ is increasing in $x$ and decreasing in $y$,
    \item $\frac{\partial g(x, y)}{\partial x} \leq g(x, y)$,
     \footnote{\label{foot:diff}We use notation for partial derivatives, but the result also holds for non-differentiable functions, if we use subgradients, etc. In particular, the result holds for the piecewise-affine functions $g$ we obtain from solving the LP in Section \ref{sec:lp}. To keep the exposition simple, we will continue using partial derivative notation throughout the paper.}
    \item $\frac{\partial g(x, y)}{\partial y} \geq g(x, y) - 1$, and
    \item for all $x, y, y'$ with $y' > y$, we have
    $$g(1, y) - g(x, y) \geq g(1, y') - g(x, y')$$
\end{enumerate}
Then Lemma 4.1 in \cite{HTWZ19} still holds, and the competitive ratio of the RANKING algorithm is at least
\begin{equation}
\min_{0 \leq \gamma, \tau \leq 1} \biggl\{(1-\tau)(1-\gamma) + (1-\tau)\int_0^\gamma g(x, \tau) dx    +
\int_0^\tau \min_{\theta \leq \gamma} \left\{(1-g(\theta, y)) + \int_0^\theta g(x, y)dx + \int_\theta^\gamma g(x, \tau) dx\right\} dy \biggr\} \label{cr}
\end{equation}
\end{theorem}
\begin{proof}
The result in Lemma 4.1 of \cite{HTWZ19} follows entirely from facts proved in their Lemmas 3.3, 3.4, and 3.5.  We show that these lemmas continue to hold given the conditions on $g$ above.  

\begin{claim}
Claim 2.1 in \cite{HTWZ19} still holds.
\end{claim}
\begin{proof}
Claim 2.1 in \cite{HTWZ19} is precisely condition 4.
\end{proof}

\begin{claim}
Fact 3.1 in \cite{HTWZ19} still holds.
\end{claim}
\begin{proof}
This fact follows from condition 2, that is, $g(x, y)$ is increasing in $x$ and decreasing in $y$.
\end{proof}

\begin{claim}
Lemma 3.1 in \cite{HTWZ19} still holds.
\end{claim}
\begin{proof}
This lemma follows mostly from condition 2, and that part of the proof still holds. The only part of the proof which doesn't directly follow from condition 2 is the assertion that if $\theta(x) = 1$ for some $x \in [0,1]$, then $\theta(x') = 1 $ for all $x' \geq x$. This part mimics the proof in \cite{HTWZ19}, except we use condition 5 instead of the stronger assumption $g(x, y) = \frac12(1+h(x)-h(y))$. 

Suppose for the sake of contradiction that $\theta(x) = 1$ but $\theta(x') < 1$ for some $x' > x$.  Since $\theta(x') < 1$, this means when $y_u = x'$ and $y_v = 1$, we have that $v$ is unmatched when $u$ arrives, and $u$ matched to some vertex $z \neq v$. Thus, $w_z(1 - g(y_z, x')) > w_v(1-g(1, x'))$. 

Now, consider what happens when $u$ arrives at time $x$. Since $\theta(x) = 1$, this means when $y_u = x$ and $y_v = 1$, we have that $v$ is matched to $u$. On the other hand, $z$ is an unmatched neighbour of $u$ if $u$ arrives at time $x$. (This is because $z$ is unmatched if $u$ arrives at the later time $x'$.) Moreover, choosing $z$ induces utility 
\begin{align*}
w_z(1 - g(y_z, x)) 
&= w_z(1-g(y_z, x')) \cdot \frac{1-g(y_z, x)}{1-g(y_z, x')} \\
&> w_v(1 - g(1, x'))\cdot \frac{1-g(y_z, x)}{1-g(y_z, x')} \\
&\geq w_v(1 - g(1, x')) \cdot \frac{1 - g(1, x)}{1 - g(1, x')} \\
&= w_v(1 - g(1, x))
\end{align*}
The first inequality is because $u$ matched to $z$ over $v$ when $y_u = x'$. The second inequality is true because
\begin{align*}
    \frac{1-g(y_z, x)}{1-g(y_z, x')}
    &= \frac{1-g(1, x) + g(1, x) - g(y_z, x)}{1-g(1, x') + g(1, x') - g(y_z, x')}
\end{align*}
If we let $a = 1-g(1, x)$, $b=1-g(1, x')$, $c = g(1, x) - g(y_z, x)$, and $d = g(1, x') - g(y_z, x')$, then the expression above is $\frac{a+c}{b+d}$. Observe that for positive $a, b, c, d$, whenever $a \leq b$ and $c \geq d$, we have $\frac{a+c}{b+d} \geq \frac{a}{b}$. The second inequality follows directly from this observation. Note that $a \leq b$ follows from condition 2 (monotonicity of $g$), and $c \geq d$ is condition 5.

Thus we have shown that $w_z(1-g(y_z, x)) > w_v(1-g(1,x))$, which says that $u$ is better off choosing $z$ than $v$. This is a contradiction. 
\end{proof}
\begin{remark}
This is the only place where condition 5 is used. 
\end{remark}

\begin{claim}
Lemma 3.3 in \cite{HTWZ19} still holds.
\end{claim}
\begin{proof}
This lemma relies on nothing but the definition of $\gamma$ and $\tau$. 
\end{proof}

\begin{claim}
Lemma 3.4 in \cite{HTWZ19} still holds.
\end{claim}
\begin{proof}
The proof of Lemma 3.4 has three parts:
\begin{enumerate}
    \item In the first part, they show that for any fixed $y_v = x < \gamma$, we have $\alpha_v \geq w_v\cdot g(x, \beta^{-1}(x))$ for all $y_u \in [0,1]$.
    \item Next, they show that $\frac{1}{w_v}\cdot \E_{y_u}\left[\alpha_v\cdot \one(y_v < \gamma) + \alpha_u\cdot \one(y_v<\gamma, y_u > \tau)\right] \geq f(x, \beta^{-1}(x))$, where
    $$f(x, \beta^{-1}(x)) := g(x, \beta^{-1}(x)) + \max\{0, \beta^{-1}(x)-\tau\}\cdot(1-g(x, \beta^{-1}(x)))$$
    \item Finally, they show that $f(x, \beta^{-1}(x)) \geq g(x, \tau)$. 
\end{enumerate} The proofs of 1 and 2 follow entirely from Lemma 3.1 in \cite{HTWZ19}, which in turn follows from our conditions 2 and 5. More specifically, the proof of 1 only uses the monotonicity of $g$ (condition 2), whereas the proof of 2 also relies on the fact if $\theta(x) = 1$ then $\theta(x') = 1$ for all $x' \geq x$, so it uses condition 5 as well.  

The proof of 3 follows unchanged, and uses condition 2 and condition 4. For completeness we write the argument here:
\begin{itemize}
    \item If $\beta^{-1}(x) < \tau$, then $f(x, \beta^{-1}(x)) = g(x, \beta^{-1}(x)) \geq g(x, \tau)$, since $\frac{\partial g(x, y)}{\partial y} \leq 0$. 
    \item If $\beta^{-1}(x) \geq \tau$, then $f(x, \beta^{-1}(x)) = g(x, \beta^{-1}(x)) + (\beta^{-1}(x) - \tau)(1 - g(x, \beta^{-1}(x))$. Observe that $f(x, \beta^{-1}(x))$ is non-decreasing in its second argument, since
    $$\frac{\partial f(x, \beta^{-1}(x))}{\partial \beta^{-1}(x)} = \underbrace{\frac{\partial g(x, \beta^{-1}(x))}{\partial \beta^{-1}(x)}  + 1 - g(x, \beta^{-1}(x))}_{\text{$\geq 0$ by condition 4}} - \underbrace{(\beta^{-1}(x) - \tau)}_{\geq 0}\cdot\underbrace{\frac{\partial g(x, \beta^{-1}(x))}{\partial \beta^{-1}(x)}}_{\text{$\leq 0$ by condition 2}} \geq 0.$$
    Thus, $f(x, \beta^{-1}(x)) \geq f(x, \tau) = g(x, \tau)$. 
\end{itemize}
\end{proof}

\begin{claim}
Lemma 3.5 in \cite{HTWZ19} still holds.
\end{claim}
\begin{proof}
In this lemma, the authors write $w_v\cdot g(y_u, y_v)$ to denote the gain, $\alpha_u$, of the online vertex $u$. This relies on their assumption that $g(y_v, y_u) + g(y_u, y_v) = 1$, because the RANKING algorithm actually sets $\alpha_u = w_v\cdot(1 - g(y_v, y_u))$.  Since we are not assuming $g(y_v, y_u) + g(y_u, y_v) = 1$, we need to replace all occurrences of the form $g(y_u, y_v)$ in their lemma with $1- g(y_v, y_u)$. We show that everything in the lemma goes through with this replacement. With this replacement, the statement of the lemma becomes
$$\E\left[\alpha_u\cdot\one(y_u < \tau) + \alpha_v \cdot \one(y_u<\tau, y_v > \gamma)\right] \geq w_v\cdot\int_0^\tau (1-g(\gamma, x)) dx.$$

The proof of Lemma 3.5 has three parts (again, the statements below are rewritten from their proof, replacing all occurrences of the form $g(y_u, y_v)$ with $1-g(y_v, y_u)$):
\begin{enumerate}
    \item In the first part, they show that for any fixed $y_u = x < \tau$, we have $\alpha_u \geq w_v\cdot (1-g(\theta(y_u), y_u))$
    \item Next, they show that for fixed $y_u = x < \tau$, $\frac{1}{w_v}\cdot \E_{y_v}\left[\alpha_u\cdot\one(y_u< \tau) + \alpha_v\cdot\one(y_u<\tau, y_v>\gamma)\right] \geq f(\theta(x), x)$, where
    $$f(\theta(x), x) := (1-g(\theta(x), x)) + \max\{0, \theta(x)-\gamma\}\cdot g(\theta(x), x)$$
    \emph{Note: In \cite{HTWZ19}, they use the notation $f(x, \theta(x))$. We choose to use $f(\theta(x), x)$ here in order to make it consistent with the notation $g(\theta(x), x)$.}
    \item Finally, they show that $f(\theta(x), x) \geq 1 - g(\gamma, x)$. 
\end{enumerate}

The proofs of 1 and 2 follow entirely from Lemma 3.1 in \cite{HTWZ19}, which in turn follows from our conditions 2 and 5. Actually, the proofs of 1 and 2 only rely on the part of Lemma 3.1 which follow from the monotonicity of $g$. (i.e. The parts of Lemma 3.1 except the statement that $\theta(x) = 1$ for some $x$ implies $\theta(x') = 1$ for all $x' \geq x$.) Thus 1 and 2 follow entirely from our condition 2, and do not need to use condition 5. 

The proof of 3 uses, in addition, condition 3. The proof in \cite{HTWZ19} follows through unchanged, but for completeness we write the argument here:
\begin{itemize}
    \item If $\theta(x) \leq \gamma$, then $f(\theta(x), x) = 1 - g(\theta(x), x) \geq 1 - g(\gamma, x)$, because $g$ is increasing in its first argument. (Condition 2.)
    \item If $\theta(x) > \gamma$, then $f(\theta(x), x)$ is non-decreasing in its first argument, since
    $$\frac{\partial f(\theta(x), x)}{\partial \theta(x)} = \underbrace{-\frac{\partial g(\theta(x), x)}{\partial \theta(x)} + g(\theta(x), x)}_{\text{$\geq 0$ by condition 3}} + \underbrace{(\theta(x) - \gamma)}_{\geq 0}\cdot\underbrace{\frac{\partial g(\theta(x), x)}{\partial \theta(x)}}_{\text{$\geq 0$ by condition 2}} \geq 0 $$
    Therefore $f(\theta(x), x) \geq f(\gamma, x) = 1 - g(\gamma, x)$. 
\end{itemize}

\end{proof}

\end{proof}
From now on, we will refer to the five conditions in Theorem \ref{thm:main_bound} as conditions 1-5.

%% file: lp.tex
To find a function $g$ that maximizes the bound in Theorem \ref{thm:main_bound}, we discretize $[0,1]^2$ into an $n \times n$ grid for a sufficiently large positive integer $n$, and write an LP to search for the values of $g$ on this discretized grid.

In Section \ref{sec:lp_constraints}, we formulate the conditions 1-5, which are the conditions that any feasible $g$ must satisfy, as constraints in the LP. Next, in Section \ref{sec:lp_objective}, we formulate the expression in Theorem \ref{thm:main_bound}, which is the bound we are trying to maximize, as an LP objective. Finally, in Section \ref{sec:lp_extend}, we will see how to extend the values of $g$ on the discretized $n \times n$ grid, which is what the LP returns, to a function $g$ defined on the entire unit square. 

\subsection{Formulating the Constraints}
\label{sec:lp_constraints}
 In this section, we show how to formulate the conditions 1-5 as constraints in the LP. 

Fix a positive integer $n$ and let $x_i = y_i = \frac{i}{n}$, for $i= 0, 1,  \ldots, n$. Our LP will have variables $g(x_i, y_j)$, the values of $g$ on the discretized unit square. Next, we encode the conditions 1-5 as constraints of the LP. Below are the conditions, and their corresponding LP constraints:
\begin{enumerate}
    \item \underline{$g(x, y): [0,1]^2 \to [0,1]$ and $g$ is continuous.} The corresponding LP constraints are $0 \leq g(x_i, y_j) \leq 1$, for all $i,j = 0, 1, \ldots, n$. Note that we do not include any constraints to enforce the continuity of $g$, since the aim of the LP is to determine the value of $g$ at a discretized set of points.
    \item \underline{$g(x, y)$ is increasing in $x$ and decreasing in $y$.} The corresponding LP constraints are 
    \begin{itemize}
        \item $g(x_i, y_j) \leq g(x_k, y_j)$ for all $0 \leq i, j, k \leq n$ with $i \leq k$;
        \item $g(x_i, y_j) \geq g(x_i, y_l)$ for all $0 \leq i, j, l \leq n$ with $j \leq l$.
    \end{itemize} 
    \item \underline{$\frac{\partial g(x, y)}{\partial x} \leq g(x, y)$}. We discretize this constraint to create the following LP constraints:
    \begin{itemize}
    \item $\frac{g(x_{i+1}, y_j) - g(x_i, y_j)}{x_{i+1} - x_i} \leq g(x_i, y_{j+1}) \quad \text{for all $0 \leq i, j \leq n-1$}$
    \item $\frac{g(x_{i+1}, y_n) - g(x_i, y_n)}{x_{i+1} - x_i} \leq g(x_i, y_n) \quad \text{for all $0 \leq i \leq n-1$}$
    \end{itemize}
    
    \begin{remark*}
    It is more natural to encode the constraints as $\frac{g(x_{i+1}, y_j) - g(x_i, y_j)}{x_{i+1} - x_i} \leq g(x_i, y_{j})$ for all $0 \leq i \leq n-1$, $0 \leq j \leq n$. Since $g(x_i, y_{j+1}) \leq g(x_i, y_j)$, our constraints are even stronger. We do this because when we extend $g$ from its discretized values to a function defined on the entire unit square, this slightly stronger version of the constraint will be needed to show that the extended function also satisfies the condition. 
    \end{remark*}
    
    \item \underline{$\frac{\partial g(x, y)}{\partial y} \geq g(x, y) - 1$}. As with the previous constraint, the corresponding LP constraints are
    \begin{itemize}
        \item $\frac{g(x_i, y_{j+1})-g(x_i, y_j)}{y_{j+1}-y_j} \geq g(x_{i+1}, y_j) - 1$, for all $0 \leq i, j \leq n-1$
        \item $\frac{g(x_n, y_{j+1})-g(x_n, y_j)}{y_{j+1}-y_j} \geq g(x_n, y_j) - 1$, for all $0 \leq  j \leq n-1$. 
    \end{itemize}
    \item \underline{For all $x, y, y'$ with $y' > y$,
    $g(1, y) - g(x, y) \geq g(1, y') - g(x, y')$}. The corresponding LP constraints are
    $$g(x_n, y_j) - g(x_i, y_j) \geq g(x_n, y_l) - g(x_i, y_l) \quad \text{for all $0 \leq i,j \leq n$ with $l > j$.}$$
\end{enumerate}

\subsection{Formulating the Objective}
\label{sec:lp_objective}
The expression we are trying to maximize is given in (\ref{cr}).
To formulate this approximately as an LP objective, we 
\begin{enumerate}
    \item Approximate the $\min_{0\leq \gamma, \tau \leq 1}$ and $\min_{\theta \leq \gamma}$ expressions by minimizing over a finite set of values, and 
    \item Approximate the integrals by finite sums. 
\end{enumerate}
We begin by letting $f(\gamma, \tau)$ be the expression inside the outermost $\min$, so that the bound is equal to $\min_{0 \leq\gamma, \tau \leq 1} f(\gamma, \tau)$. Since we cannot check all values of $\gamma$ and $\tau$, we approximate it by $\min_{0 \leq i, j \leq n} f(x_i, y_j)$.  We write this as a linear objective using the standard trick of introducing a dummy variable $t$, and writing
\begin{align*}
    \max \quad &t \\
    \text{s.t.} \quad &t \leq f(x_i, y_j) \quad \text{for all $0 \leq i, j \leq n$.}
\end{align*}
Next, we must write constraints to model $f(x_i, y_j)$. We replace the inner $\min_{\theta \leq x_i}$ by a minimum over the discretized grid: $\min_{\theta \leq x_i}$ becomes $\min_{x_k \leq x_i}$. For each integral that appears in the expression for $f$, we replace it by a left Riemann sum. For example, the integral $\int_0^{x_i} g(x, y_j) dx$ would be replaced by $\frac{1}{n} \sum_{k=0}^{i-1} g(x_k, y_j)$. 

With these approximations, we can approximate $f(x_i, y_j)$ as a linear function $\tilde{f}(x_i, y_j)$ of the $g(x_i, y_j)$ variables:
\begin{align*}
&f(x_i, y_j) \approx \tilde{f}(x_i, y_j) =  (1-x_i)(1-y_j) + (1-y_j)\cdot \frac1n \sum_{k=0}^{i-1}g(x_k, y_j) \\
&+ \frac{1}{n}\sum_{l=0}^{j-1} \min_{k \leq i} \left\{(1-g(x_k, y_l)) + \frac{1}{n}\sum_{d=0}^{k-1} g(x_d, y_l)  + \frac1n \sum_{d=k}^{i-1} g(x_d, y_j)\right\}
\end{align*}

Hence, to summarize this section and Section \ref{sec:lp_constraints}, the full linear program we use the compute the values of $g$ on the discretized $n \times n$ grid is as follows:
\begin{align*}
    \max \quad &t \\
    \text{s.t.} \quad &t \leq \tilde{f}(x_i, y_j) \quad \text{for all $0 \leq i, j \leq n$} \\
    &\text{and such that $g$ satisfies the constraints from Section \ref{sec:lp_constraints}.}
\end{align*}

\subsection{Extending the Discretized Function to the Unit Square}
\label{sec:lp_extend}
The linear program gives us values of $g$ on any given discretization of $[0,1]^2$, but to use the bound in Theorem \ref{thm:main_bound} we must 
\begin{enumerate}
    \item Extend $g$ to be defined on the entire unit square, and
    \item Show that this extended function satisfies Conditions 1-5. 
\end{enumerate}
To extend $g$ from its values on an $n \times n$ grid to a function defined on the entire unit square, we triangulate the $n \times n$ grid as shown in Figure \ref{fig:triangulation}.

\begin{figure}
\begin{minipage}[b]{0.30\textwidth}
    \centering
    \begin{tikzpicture}[scale=0.6]
    \draw[step=1.0,black] (0, 0) grid (4, 4);
    \draw[thin] (0, 1) -- (1, 0);
    \draw[thin] (0, 2) -- (2, 0);
    \draw[thin] (0, 3) -- (3, 0);
    \draw[thin] (0, 4) -- (4, 0);
    \draw[thin] (1, 4) -- (4, 1);
    \draw[thin] (2, 4) -- (4, 2);
    \draw[thin] (3, 4) -- (4, 3);
    \draw[->, thick] (0,0)--(4.5,0) node[right]{$x$};
    \draw[->, thick] (0,0)--(0,4.5) node[above]{$y$};
    \end{tikzpicture}
    \caption{Triangulating the grid. Here, $n = 4$.}
    \label{fig:triangulation}
    \end{minipage} \hfill
\begin{minipage}[b]{0.30\textwidth}
    \centering    
  \begin{tikzpicture}[scale=0.6]
    \vertex (x1) at (0, 0)  [label=below left:{$(a_1, b_1)$}] {};
    \vertex (x2) at (0, 3)  [label={$(a_2, b_2)$}] {};
    \vertex (x3) at (3, 0)  [label=below right:{$(a_3, b_3)$}] {};
    \vertex (p) at (1, 1) [label=below:{$(x, y)$}] {};
    \draw (x1) -- (x2) -- (x3) -- (x1);
    \end{tikzpicture}
    \caption{Extending the function values to a point inside a triangle.}
    \label{fig:triangle}
\end{minipage} \hfill
\begin{minipage}[b]{0.30\textwidth}
       \centering
        \begin{tikzpicture}[scale=0.6]
    \vertex (x1) at (0, 0)  [label= left:{$(x_i, y_j)$}] {};
    \vertex (x2) at (3, 0)  [label= right:{$(x_{i+1}, y_j)$}] {};
    \vertex (y1) at (3, -3)  [] {};
    \vertex (y2) at (0, 3) [] {};
    \draw (x1) -- (x2) -- (y1) -- (x1);
    \draw(x1) -- (x2) -- (y2) -- (x1);
    
    \node at (1.5, -0.25) {$l$};
    \node at (1,1) {$T_2$};
    \node at (2,-1) {$T_1$};
    \end{tikzpicture}
        \caption{Illustration of the proof of Condition 3.}
        \label{fig:cond3}
\end{minipage}
\end{figure}
For a point $(x, y)$ on a gridpoint, its function value is given by the LP. For any other point $(x, y)$, we define $g(x, y)$ to be a convex combination of the function values on the three vertices of the triangle containing $(x, y)$. More precisely, suppose $(x, y)$ is contained in the triangle with vertices $(a_1, b_1), (a_2, b_2)$, and $(a_3, b_3)$, where the $(a_i, b_i)$ are gridpoints. (See Figure \ref{fig:triangle}.) Then we define 
$$g(x, y) = \lambda_1 \cdot g(a_1, b_1) + \lambda_2 \cdot g(a_2, b_2) + \lambda_3 \cdot g(a_3, b_3),$$
where $\lambda_1, \lambda_2, \lambda_3$ are the unique coefficients that satisfy $\lambda_1, \lambda_2, \lambda_3 \geq 0$,  $\lambda_1 + \lambda_2 + \lambda_3 = 1$ and 
$$(x, y) = \lambda_1 \cdot(a_1, b_1) + \lambda_2 \cdot (a_2, b_2) + \lambda_3 \cdot (a_3, b_3).$$
Geometrically, the extended function is piecewise affine -- it is affine on each triangle. 

We devote the remainder of this section to proving that the extended function satisfies conditions 1-5. We list the conditions below, and prove that the extended function satisfies them:

\begin{enumerate}
    \item \underline{$g(x, y): [0,1]^2 \to [0,1]$ and $g$ is continuous.} The extended function takes values in $[0,1]$ because its values are convex combinations of its values on the discretized grid, which are in $[0,1]$. It is continuous because it is piecewise affine. 
    
    \item \underline{$g(x, y)$ is increasing in $x$ and decreasing in $y$.} We will show that $g$ is increasing in $x$; the proof that it is decreasing in $y$ is similar.  
    
    Let $(a_1, b)$ and $(a_2, b)$ be points in the unit square, with $a_1 \leq a_2$. We must show that $g(a_1, b) \leq g(a_2, b)$. First, observe that it suffices to show this when the two points are contained in the same triangle. This is because if $(a_1, b)$ and $(a_2, b)$ were contained in different triangles, then the horizontal line segment $l$ from $(a_1, b)$ to $(a_2, b)$ can be divided from left to right into a sequence of segments (say $l_1, \ldots, l_k$), each of which is contained in a single triangle. Then the fact that $g$ is increasing on each smaller segment $l_i$ would imply that $g$ is increasing on $l$. (For an illustration of this, see Figure
    
    \begin{figure}[t!]
    \centering
    \begin{subfigure}[t]{0.45\textwidth}
        \centering
        \begin{tikzpicture}[]

    \draw[step=2.0,black, thick] (0, -2) grid (6, 2);
    \foreach \i in {0,...,2} {
    \draw[thick] (2*\i, 2) -- (\i*2+2, 0);
    }
    \foreach \i in {0,...,2} {
    \draw[thick] (2*\i, 0) -- (\i*2+2, -2);
    }
    
    \vertex (x1) at (0.55, 0.5) [label = left:{$(a_1, b)$}] {};
    \vertex (x2) at (4.7, 0.5) [label = right:{$(a_2, b)$}] {};
    \vertex (a1) at (1.5, 0.5) {};
    \vertex (a1) at (2, 0.5) {};
    \vertex (a1) at (3.5, 0.5) {};
    \vertex (a1) at (4, 0.5) {};
    \draw[thin] (x1) -- (x2);

    \end{tikzpicture}
        \caption{Illustration of the proof of why we may assume $(a_1, b)$ and $(a_2, b)$ are contained in the same triangle.}
    \end{subfigure}%
    ~ 
    \begin{subfigure}[t]{0.45\textwidth}
        \centering
        \begin{tikzpicture}
        
        \vertex (x0) at (0, 0) [label=left:{$(0,0)$}] {};
        \vertex (x1) at (4, 0) [label=right:{$(x_1,0)$}] {};
        \vertex (x2) at (0, 4) [label= right:{$(0,y_1)$}] {};
        \vertex (a1) at (1, 1) [label=below:{$(a_1, b)$}] {};
        \vertex (a2) at (2, 1) [label=below:{$(a_2, b)$}] {};
        \vertex (y0) at (0, 1) [label=left:{$(0, b)$}] {};
        \vertex (y1) at (3, 1) [label=right:{$(\frac1n - b, b)$}] {};

    \draw[thick] (x0) -- (x1) -- (x2) -- (x0);
    \draw (y0) -- (a1) -- (a2) -- (y1);
    
        \end{tikzpicture}
        \caption{Illustration of the proof of why $g(a_1, b) \leq g(a_2, b)$ for $(a_1, b)$, $(a_2, b)$ in the same triangle and $a_1 \leq a_2$.}
    \end{subfigure}
    \caption{Side-by-side comparison of the function $g$ we used, versus the function $g$ used by Huang et al.}
    \label{fig:lp_extend_monotone}
\end{figure}
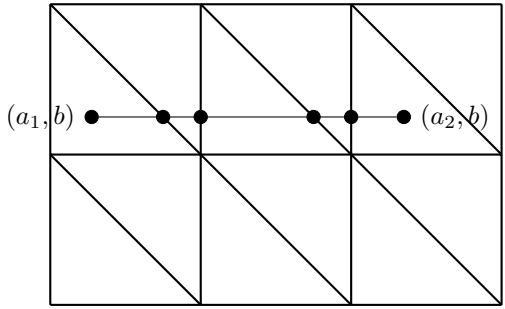
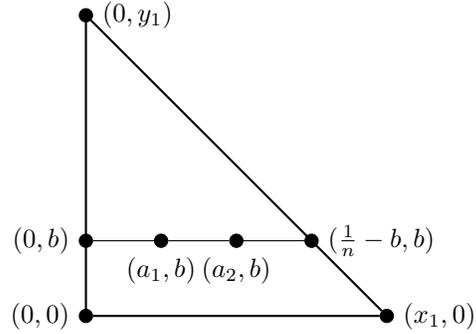
    
    So, we can assume $(a_1, b)$ and $(a_2, b)$ are contained in the same triangle. Without loss of generality, suppose $(a_1, b)$ and $(a_2, b)$ are both contained in the lower-leftmost triangle; that is, the triangle with vertices $(0, 0)$, $(x_1, 0)$, and $(0, y_1)$; the proof for any other triangle is the same. 
    
    Note that $(a_1, b)$ and $(a_2, b)$ are both on the line segment from $(0, b)$ to $(\frac{1}{n}-b, b)$. Since $g$ is piecewise affine in any triangle, it follows that $g(a_1, b) = (1-\lambda_1)\cdot g(0, b) + \lambda_1 \cdot  g(\frac1n - b, b)$, where $0 \leq \lambda_1 \leq 1$ satisfies $\lambda_1 \cdot (\frac1n - b) = a_1$. Similarly, $g(a_2, b) = (1-\lambda_2)\cdot g(0, b) + \lambda_2 \cdot g(\frac1n - b, b)$, where $0 \leq \lambda_2 \leq 1$ satisfies $\lambda_2\cdot(\frac1n - b) = a_2$. Now, since $a_1 \leq a_2$, it follows that $\lambda_1 \leq \lambda_2$. Therefore, to show that $g(a_1, b) \leq g(a_2, b)$ it suffices to show that $g(0, b) \leq g(\frac1n - b, b)$. 
    
    To see this, we note that $g(0, b) = (1-\lambda) \cdot g(0, 0) + \lambda \cdot g(0, \frac1n)$, where $0 \leq \lambda \leq 1 $ satisfies $\frac{\lambda}{n} = b$. Similarly, $g(\frac1n - b, b) = (1-\lambda)\cdot g(\frac1n, 0) + \lambda \cdot g(0, \frac1n)$. Since $g(\frac1n, 0) \geq g(0, 0)$ (this was a constraint in the LP), it follows that $g(0, b) \leq g(\frac1n - b, b)$, as needed.
    
    \item \underline{$\frac{\partial g(x, y)}{\partial x} \leq g(x, y)$}.  Consider a horizontal line segment $l$ between two adjacent gridpoints, say between $(x_i, y_j)$ and $(x_{i+1}, y_j)$. In the triangulation, $l$ is adjacent to two triangles: one triangle $T_1$ below it and one triangle $T_2$ above it. (If $y_i = 0$ or $y_i = 1$, then $l$ is only adjacent to one triangle, but the same argument still goes through.) See Figure \ref{fig:cond3} for an illlustration. Because $g$ is piecewise affine in each triangle, it follows that $\frac{\partial g(x, y)}{\partial x}$ is constant on $T_1 \cup T_2$, and is equal to the slope of $l$. Recall that the LP imposes the following constraint on the slope of $l$:
    $$\mathrm{slope}(l) = \frac{g(x_{i+1}, y_j) - g(x_i, y_j)}{x_{i+1} - x_i} \leq g(x_i, y_{j+1})$$
    Because $g$ is increasing in $x$ and decreasing in $y$, we note that $g(x_i, y_{j+1}) \leq \inf\{g(x, y): (x, y) \in T_1 \cup T_2\}$. Thus $\frac{\partial g(x, y)}{\partial x} \leq g(x, y)$ holds on $T_1 \cup T_2$. Because any triangle is adjacent to some horizonal line segment in the grid, this argument shows that $\frac{\partial g(x, y)}{\partial x} \leq g(x, y)$ holds for all $(x, y)$ in the unit square, and we are done. 
    
    \item \underline{$\frac{\partial g(x, y)}{\partial y} \geq g(x, y) - 1$}. The proof of this is similar to the proof of the previous condition. 
    
    \item \underline{For all $x, y, y'$ with $y' > y$,
    $g(1, y) - g(x, y) \geq g(1, y') - g(x, y')$}.
    Let $\sF = \{0, \frac1n, \frac2n, \ldots, 1\}$. If $x,  y, y' \in \sF$, then the condition holds, because these were constraints imposed by the LP.
    
    Suppose now $(x, y)$ lies in the interior of some triangle $T$. Fix $x$ and $y'$, and imagine varying $y$ up and down such that $(x, y)$ remains inside $T$. Let $I$ be the range of values of $y$ such that $(x, y)$ remains inside $T$. Since $g$ is affine on each triangle, it follows that $\frac{\partial}{\partial y} \left(g(1, y) - g(x, y)\right)$ is constant for all $y$ in $I$. Therefore (by moving $y$ in the direction that decreases the LHS of the inequality if necessary), it suffices to prove the inequality in the case $(x, y)$ is on the boundary of a triangle. Similarly, we may assume that $(x, y')$ lies on the boundary of a triangle. 
    
    Suppose $(x, y)$ and $(x, y')$ both lie on hypotenuses (see Figure \ref{fig:hypo}). The case where one or both of the points lie on a base of a triangle is very similar (and easier), so we will omit it here. 
    
    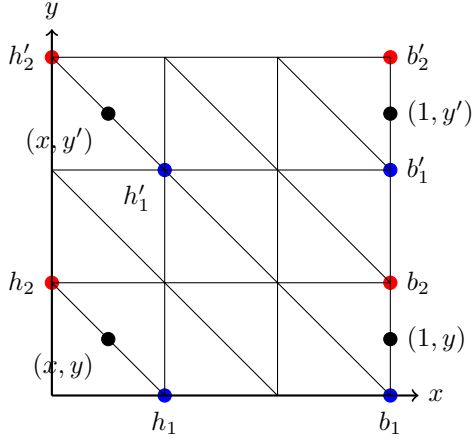
\begin{figure}
    \centering
    \begin{tikzpicture}[scale=.75]
    \vertex (x1) at (1, 1) [label=below left:{$(x, y)$}] {};
    \vertex (y1) at (6, 1) [label=right:{$(1, y)$}] {};
    
    \vertex (x2) at (1, 5) [label=below left:{$(x, y')$}] {};
    \vertex (y2) at (6, 5) [label=right:{$(1, y')$}] {};
    
    \vertex (a1) at (2, 0) [label=below:{$h_1$}, blue] {};
    \vertex (a3) at (0, 2) [label=left:{$h_2$}, red] {};
    
    \vertex (a1') at (2, 4) [label=below left:{$h'_1$}, blue] {};
    \vertex (a3') at (0, 6) [label=left:{$h'_2$}, red] {};
    
    \vertex (b1) at (6, 0) [label=below:$b_1$, blue] {};
    \vertex (b2) at (6, 2) [label=right:$b_2$, red] {};
    
    \vertex (b1') at (6, 4) [label=right:$b'_1$, blue] {};
    \vertex (b2') at (6, 6) [label=right:$b'_2$, red] {};
    
    \draw[step=2.0,black] (0, 0) grid (6, 6);
    \foreach \i in {1,...,3} {
    \draw[thin] (0, \i*2) -- (\i*2, 0);
    }
    \foreach \i in {1,...,2} {
    \draw[thin] (\i*2, 6) -- (6, \i*2);
    }
    \draw[->, thick] (0,0)--(6.5,0) node[right]{$x$};
    \draw[->, thick] (0,0)--(0,6.5) node[above]{$y$};
    \end{tikzpicture}
    \caption{Illustration of proof of condition 5.}
    \label{fig:hypo}
    \end{figure}
    
    Let $h_1$ and $h_2$ be the two endpoints of the hypotenuse containing $(x, y)$, with $h_1$ lower than $h_2$. Similarly, define $h'_1$ and $h'_2$.  Let $b_1$ and $b_2$ be the two endpoints of the vertical grid segment containing $(1, y)$. Similarly, define $b'_1$ and $b'_2$.  We will use the fact that the inequality holds for the gridpoints $(b_1, h_1, b'_1, h'_1)$ and the gridpoints $(b_2, h_2, b'_2, h'_2)$ to deduce that it holds for our points. 
    
    The inequality on the points $(b_1, h_1, b'_1, h'_1)$ is
    $$g(b_1) - g(h_1) \geq g(b'_1) - g(h'_1)$$
    The inequality on the points $(b_2, h_2, b'_2, h'_2)$ is 
    $$g(b_2) - g(h_2) \geq g(b'_2) - g(h'_2)$$
    Now let $0 \leq \lambda \leq 1$ be the scalar so that $\lambda b_1 + (1-\lambda)b_2 = (1, y)$. Observe that we also have $\lambda h_1 + (1-\lambda)h_2 = (x, y)$, $\lambda b'_1 + (1-\lambda)b'_2 = (1, y')$, and $\lambda h'_1 + (1-\lambda) h'_2 = (x, y')$.
    
    Now, multiply the inequality for $(b_1, h_1, b'_1, h'_1)$ by $\lambda$, and multiply the inequality for $(b_2, h_2, b'_2, h'_2)$ by $(1-\lambda)$, then add them together. The result is the inequality
    $$g(1, y) - g(x, y) \geq g(1, y') - g(x, y'),$$
    which is what we wanted.
\end{enumerate}

%% file: error.tex
The linear program gives us function values defined on a discretization of the unit square, which we then extend to a function $g$ defined on the entire unit square via triangulation. It remains now to plug this $g$ into the bound for the competitive ratio given by Theorem \ref{thm:main_bound}. We cannot evaluate the bound analytically for the function $g$ returned by the LP; instead, we evaluate it computationally. 

For $0 \leq \gamma, \tau \leq 1$, let 
$$f(\gamma, \tau) =  (1-\tau)(1-\gamma) + (1-\tau)\int_0^\gamma g(x, \tau) dx + 
\int_0^\tau \min_{\theta \leq \gamma} \left\{(1-g(\theta, y)) + \int_0^\theta g(x, y)dx + \int_\theta^\gamma g(x, \tau) dx\right\} dy$$
so that, by Theorem \ref{thm:main_bound}, the competitive ratio of $g$ is at least $\min_{0 \leq \gamma, \tau \leq 1} f(\gamma, \tau)$. 

When we evaluate this bound using a computer, we incur two sources of error:
\begin{enumerate}
    \item The bound takes a minimum over all $(\gamma, \tau)$ in the unit square. However, using a computer, we can only check a finite number of points $(\gamma, \tau)$. 
    \item For a fixed $(\gamma, \tau)$, we do not calculate $f(\gamma, \tau)$ exactly. Instead, using a computer, we calculate an approximation $\hat{f}(\gamma, \tau)$, by
    \begin{itemize}
        \item Approximating the integrals with finite sums, and
        \item Replacing the inner minimum over all $\theta \leq \gamma$ by a minimum over a finite set of $\theta$. 
    \end{itemize} 
\end{enumerate}
In what follows, we will bound the errors above. This proves that the output of the computer program is a valid bound on the competitive ratio. We will show that $f$ is Lipschitz in $\gamma$ and $\tau$, which implies that checking all values of $(\gamma, \tau)$ in a sufficiently fine discretization of the unit square is enough to obtain a quantifiable bound on the error. 

Before we move on, we remind the reader what it means for a function to be Lipschitz.
\begin{defn}
A function $f: \R^n\to\R$ is $L$-Lipschitz if $\abs{f(x) - f(y)}\leq L\norm{x - y}$ for all $x, y \in \R^n$. 
\end{defn}
It will be convenient for us to work with Lipschitzness in a particular coordinate.
\begin{defn}
A function $f: \R^n \to \R$ is $L$-Lipschitz in its $i$th coordinate if
$$\abs{f(x_1, \ldots, x_i,\ldots, x_n) - f(x_1, \ldots, x'_i, \ldots, x_n)} \leq L\abs{x_i - x_i'}$$
for all $x_1, \ldots, x_i, x'_i, \ldots, x_n \in \R$. 
\end{defn}

Throughout the proofs below, we will repeatedly use the following facts:
\begin{itemize}
    \item $0 \leq \frac{\partial g(x, y)}{\partial x} \leq  g(x, y) \leq 1$ . This follows from conditions 1, 2, and 3. 
    \item $-1 \leq g(x, y) - 1 \leq \frac{\partial g(x, y)}{\partial y} \leq 0$. This follows from conditions 1, 2, and 4. 
\end{itemize}


\begin{restatable}{lemma}{approx}
\label{lem:approx1}
$f(\gamma, \tau)$ is 1-Lipschitz in $\gamma$ and $3$-Lipschitz in $\tau$.
\end{restatable}

\begin{proof}
First, consider varying $\gamma$. For $ \delta > 0$, we have
$$p(\gamma + \delta, \tau) - p(\gamma, \tau) = -\delta(1-\tau) + (1-\tau)\int_\gamma^{\gamma + \delta} \underbrace{g(x, \tau)}_{\in [0,1]} dx \in [-\delta(1-\tau), 0]$$
Also, note that
$$\frac{\partial}{\partial\theta} h(\gamma, \tau, \theta, y) = \underbrace{-\frac{\partial g(\theta, y)}{\partial \theta}}_{\in [-g(\theta, y), 0]} + g(\theta, y) - g(\theta,\gamma) \in [-g(\theta, \gamma), g(\theta, y)] \subset [-1, 1]$$
and
$$h(\gamma+\delta, \tau, \theta, y) - h(\gamma, \tau, \theta, y) = \int_\gamma^{\gamma + \delta} g(x, \tau) dx \in [0, \delta]
$$
Therefore, 
\begin{align*}
    q(\gamma+\delta, \tau, y) - q(\gamma, \tau, y) 
    &= \min_{\theta \leq \gamma + \delta} h(\gamma + \delta, \tau, \theta, y) - \min_{\theta \leq \gamma} h(\gamma, \tau, \theta, y) \\
    &\in \min_{\theta \leq \gamma + \delta} \{ h(\gamma, \tau, \theta, y) + [0, \delta] \} - \min_{\theta \leq \gamma} h(\gamma, \tau, \theta, y) \\
    &= [0, \delta] + \min_{\theta \leq \gamma + \delta}  h(\gamma, \tau, \theta, y) - \min_{\theta \leq \gamma} h(\gamma, \tau, \theta, y) \\
    &\subseteq [0, \delta] + [-\delta, 0] \\
    &= [-\delta, \delta]
\end{align*}
Combining, we get
\begin{align*}
    f(\gamma + \delta, \tau) - f(\gamma, \tau) &= p(\gamma + \delta, \tau) - p(\gamma, \tau) + \int_0^\tau (q(\gamma + \delta, \tau, y) - q(\gamma, \tau, y)) dy \\
    &\in [-\delta(1-\tau), 0] + [-\delta\tau, \delta \tau] \\
    &= [-\delta, \delta\tau] \\
    &\subseteq [-\delta, \delta]
\end{align*}
This shows that $f$ is 1-Lipschitz in $\gamma$. 

Next, we consider varying $\tau$. Let $\delta > 0$. First, we have
\begin{align*}
\frac{\partial p(\gamma, \tau)}{\partial \tau} &= -(1-\gamma) - \int_0^\gamma \underbrace{g(x, \tau)}_{\in [0,1]} dx + (1-\tau)\int_0^\gamma \underbrace{\frac{\partial g(x, \tau)}{\partial \tau}}_{\in [-1, 0]} dx \\
&\in \{-(1-\gamma)\} + [-\gamma, 0] + [-(1-\tau), 0] \\
&\subset [-2, -1]
\end{align*}
Thus, $p(\gamma, \tau + \delta) - p(\gamma, \tau) \in [-2\delta, -\delta]$. 

Next, note that 
$$h(\gamma, \tau + \delta, \theta, y) - h(\gamma, \tau, \theta, y) = \int_\theta^\gamma \underbrace{\left(g(x, \tau + \delta) - g(x, \tau)\right)}_{\in [-\delta, 0]} dy \in [-\delta, 0]$$
This implies that 
\begin{align*}
    q(\gamma, \tau + \delta, y) -  q(\gamma, \tau, y) &= \min_{\theta \leq \gamma} h(\gamma, \tau + \delta, \theta, y) - \min_{\theta \leq \gamma} h(\gamma, \tau, \theta, y) \\
    &\in \min_{\theta \leq \gamma}  \left\{h(\gamma, \tau, \theta, y) + [-\delta, 0] \right\} - \min_{\theta \leq \gamma} h(\gamma, \tau, \theta, y) \\
    &= [-\delta, 0]
\end{align*}
Hence,
\begin{align*}
    \int_0^{\tau + \delta} q(\gamma, \tau + \delta, y) dy - \int_0^\tau q(\gamma, \tau, y) dy 
    &\in (\tau + \delta)[-\delta, 0] +  \int_0^{\tau + \delta} q(\gamma, \tau, y) - \int_0^\tau q(\gamma, \tau, y) dy \\
    &\subset [-\delta, 0] + \int_\tau^{\tau + \delta} \underbrace{q(\gamma, \tau, y)}_{\in [0,2]} dy \\
    &\subset [-\delta, 0] +[0, 2\delta] \\
    &= [-\delta, 2\delta]
\end{align*}
Here, $q(\gamma, \tau, y) \in [0,2]$ is because $q(\gamma, \tau, y) = \min_{\theta \leq \gamma} h(\gamma, \tau, \theta, y)$, and 
$$h(\gamma, \tau, \theta, y) = (1-g(\theta, y)) + \int_0^\theta g(x, y)dx + \int_\theta^\gamma g(x, \tau) dx \leq 1 - g(\theta, y) + g(\theta, y) + \int_\theta^\gamma g(x, \tau) dx \leq 2.$$
(Clearly, $h(\gamma, \tau, \theta, y) \geq 0$. )

Combining, we get 
\begin{align*}
    f(\gamma, \tau + \delta) - f(\gamma, \tau) &= p(\gamma, \tau + \delta) - p(\gamma, \delta) + \int_0^{\tau + \delta} q(\gamma, \tau + \delta, y) dy - \int_0^\tau q(\gamma, \tau, y) dy \\
    &\in [-2\delta, -\delta] + [-\delta, 2\delta] \\
    &= [-3\delta, \delta] \\
    &\subset [-3\delta, 3\delta]
\end{align*}
This shows that $f$ is 3-Lipschitz in $\tau$. 

\end{proof}

The preceding lemma allows us to control the error incurred from checking the bound over all $(\gamma, \tau)$ in a discretization of the unit square instead of the entire unit square. The second source of error is that for a fixed $(\gamma, \tau)$, we evaluate an approximation $\hat{f}(\gamma, \tau)$ to $f(\gamma, \tau)$, because we replace the integrals with discrete sums and the  minimization over all $\theta \leq \gamma$ with a minimization over finitely many $\theta$. The following lemma controls the second source of error.

To make notation less cluttered, let 
\begin{itemize}
    \item $p(\gamma, \tau) = (1-\gamma)(1-\tau) + (1-\tau)\int_0^\gamma g(x, \tau) dx$, 
    \item $h(\gamma, \tau, \theta, y) = (1-g(\theta, y)) + \int_0^\theta g(x, y)dx + \int_\theta^\gamma g(x, \tau) dx$
    \item $q(\gamma, \tau, y) = \min_{\theta \leq \gamma} h(\gamma, \tau, \theta, y)$
\end{itemize}
so that $f(\gamma, \tau) = p(\gamma, \tau) + \int_0^\tau q(\gamma, \tau, y) dy$. 

\begin{restatable}{lemma}{approxx}
\label{lem:approx2}
Fix $\gamma, \tau \in [0,1]$, and let $m$ be a positive integer. 
Let $\hat{f}(\gamma, \tau)$ be the approximation to $f(\gamma, \tau)$ obtained by:
\begin{itemize}
    \item Replacing the integral $\int_0^\tau q(\gamma, \tau, y) dy$ with a trapezoidal sum with subdivision length $\frac1m$,
    \item Replacing the other three integrals with left Riemann sums with subdivision length $\frac1m$, and
    \item Replacing the minimum over all $\theta \leq \gamma$ with a minimum over a discretization with subdivision length $\frac1m$.
\end{itemize}
Then $\hat{f}(\gamma, \tau) \leq f(\gamma, \tau) + \frac{5}{4m}$. 

More precisely, $\hat{f}$ is defined as follows. Define $x_k = y_k = \frac{k}{m}$ for $k = 0, 1, \ldots, m$. Let $i$ and $j$ be the integers such that $x_i \leq \gamma < x_{i+1}$, and $y_j\leq \tau < y_{j+1}$. Then
\begin{align*}
    \hat{f}(\gamma, \tau) = \hat{p}(\gamma, \tau) + \frac{1}{m}\sum_{k=0}^{j-1} \frac{\hat{q}(\gamma, \tau, y_k) + \hat{q}(\gamma, \tau, y_{k+1})}{2}
\end{align*}
where $\hat{p}$ and $\hat{q}$ are defined to be
$$\hat{p}(\gamma, \tau) = (1- \gamma)(1-\tau) + (1-\tau) \cdot \frac1m \sum_{k=0}^{i-1} g(x_k, \tau)$$
and 
$$\hat{q}(\gamma, \tau, y) = \min_{k \leq i+1} \left\{1 - g(x_k, y) + \frac{1}{m}\sum_{d = 0}^{k-1} g(x_d, y) + \frac{1}{m} \sum_{d=k}^{i-1} g(x, y_{j+1}) \right\}$$

\end{restatable}

\begin{proof}
Recall that $f(\gamma, \tau) = p(\gamma, \tau) + \int_0^\tau q(\gamma, \tau, y) dy$, where
\begin{itemize}
    \item $p(\gamma, \tau) = (1-\gamma)(1-\tau)   + (1-\tau) \int_0^\gamma g(x, \tau) dx$ and
    \item $q(\gamma, \tau, y) = \int_0^\tau \min_{\theta \leq \gamma} \left\{ 1 - g(\theta, y) + \int_0^\theta g(x, y) dx + \int_\theta^\tau g(x, \tau) dx \right\}$
\end{itemize}
First, note that left Riemann sums always underapproximate the integral of an increasing function, and $g(x, y)$ is increasing in $x$. Hence, wherever we discretize an integral of $g$ with respect to $x$, we obtain a lower bound. It follows that
\begin{itemize}
    \item $\hat{p}(\gamma, \tau) \leq p(\gamma, \tau)$, and 
    \item $\hat{q}(\gamma, \tau, y) \leq \min_{k \leq i+1} \left\{1 - g(x_k, y) + \int_0^{x_k} g(x, y) dx + \int_{x_k}^{x_i} g(x, y_{j+1}) dx \right\}$. 
    
    Let $\sF_\gamma = \{x_k: x_k \leq \gamma\}$. 
    Note that we can further upper bound the previous expression by  
    $\min_{\theta \in \sF_\gamma} \left\{1 - g(\theta, y) + \int_0^{\theta} g(x, y) dx + \int_{\theta}^{\gamma} g(x, \tau) dx \right\}$; this is because $x_i \leq \gamma < x_{i+1}$ and $y_j \leq \tau < y_{j+1}$, and $g(x, y_{j+1}) \leq g(x, \tau)$. 
\end{itemize}
Hence,   
\begin{align*}
\hat{f}(\gamma, \tau)
&\leq p(\gamma, \tau) + \frac{1}{m} \sum_{k=0}^{j-1}\frac{\min_{\theta \in \sF_\gamma} h(\gamma, \tau, \theta, y_{k}) + \min_{\theta \in \sF_\gamma} h(\gamma, \tau, \theta, y_{k+1})}{2},
\end{align*}
where $h(\gamma, \tau, \theta, y) := 1 - g(\theta, y) + \int_0^{\theta} g(x, y) dx + \int_{\theta}^{\gamma} g(x, \tau) dx$. 
The rest of the proof will be devoted to bounding the second term in the displayed inequality above.

Observe that $h$ is 1-Lipschitz in $\theta$, because 
$$\frac{\partial}{\partial \theta} h(\gamma, \tau, \theta, y) = \underbrace{-\frac{\partial}{\partial \theta} g(\theta, y)}_{\in [-g(\theta, y), 0]} + g(\theta, y) - g(\theta, \tau)  \in [-g(\theta, \tau), g(\theta, y)] \subset [-1, 1]$$
This implies that, for any $\gamma, \tau, y$,
$$\min_{\theta \in F_\gamma} h(\gamma, \tau, \theta, y) \leq \min_{\theta \leq \gamma} h(\gamma, \tau, \theta, y) + \frac{1}{m}$$
This is because if $\theta^* = \arg\min_{\theta \leq \gamma} h(\gamma, \tau, \theta, y)$, we can find some $\bar{\theta} \in F_\gamma$ with $\abs{\bar{\theta} - \theta^*} < \frac1m$, and then
$$\min_{\theta \in F_\gamma} h(\gamma, \tau, \theta, y) \leq h(\gamma, \tau, \bar{\theta}, y) \leq h(\gamma, \tau, \theta^*, y) + \frac1m = \min_{\theta \leq \gamma} h(\gamma, \tau, \theta, y) + \frac1m,$$
where the second inequality is because $h$ is 1-Lipschitz in $\theta$. 

Next, observe that $h(\gamma, \tau, \theta, y)$ is 1-Lipschitz in $y$, because
$$\frac{\partial}{\partial y} h(\gamma, \tau, \theta, y) = \underbrace{-\frac{\partial g(\theta, y)}{\partial y}}_{\in [0, 1]} + \int_0^\theta \underbrace{\frac{\partial g(x, y)}{\partial y}}_{\in [-1, 0]} dx\in [-1, 1]$$
This implies that $q(\gamma, \tau, y) = \min_{\theta \leq \gamma} h(\gamma, \tau, \theta, y)$ is 1-Lipschitz in $y$.

We can now put these facts together to prove the desired bound. We have 
\begin{align*}
    &\frac{1}{m} \sum_{k=0}^{j-1}\frac{\min_{\theta \in \sF_\gamma} h(\gamma, \tau, \theta, y_{k}) + \min_{\theta \in \sF_\gamma} h(\gamma, \tau, \theta, y_{k+1})}{2}
    \\
    &\leq \frac1m \sum_{k=0}^{j-1} \left(\frac1m + \frac{\min_{\theta \leq \gamma} h(\gamma, \tau, \theta, y_{k}) + \min_{\theta \leq \gamma} h(\gamma, \tau, \theta, y_{k+1})}{2}\right) \\
    &=  \frac{j}{m^2} + \frac1m \sum_{k=0}^{j-1} \left( \frac{q(\gamma, \tau, y_k) + q(\gamma, \tau, y_{k+1})}{2}\right) \\
    &\leq \frac{j}{m^2} + \frac{j}{4m^2} + \int_0^\tau q(\gamma, \tau, y) dy 
\end{align*}
The last inequality is because $\frac1m \sum_{k=0}^{j-1} \left( \frac{q(\gamma, \tau, y_k) + q(\gamma, \tau, y_{k+1})}{2}\right)$ is the trapezoidal sum approximation of the integral $\int_0^\tau q(\gamma, \tau, y) dy$, with discretization length $\frac1m$. Since $q(\gamma, \tau, y)$ is 1-Lipschitz in $y$, this incurs an additive error of at most $\frac{j}{4m^2}$ by Lemma \ref{lem:lipschitz_integral} in Appendix \ref{app:lipschitz}. 

Combining, we get
\begin{align*}
\hat{f}(\gamma, \tau) &\leq p(\gamma, \tau) + \frac{j}{m^2} + \frac{j}{4m^2} + \int_0^\tau q(\gamma, \tau, y) dy \\
&= f(\gamma, \tau) + \frac{j}{m^2} + \frac{j}{4m^2} \\
&\leq f(\gamma, \tau) + \frac{5}{4m},
\end{align*}
as claimed. 

\end{proof}

Combining the above two lemmas allows us to  quantify the error incurred when we evaluate the bound in Theorem \ref{thm:main_bound} using a computer.

\begin{corollary}
\label{cor:compcheck}
Let $\sF = \{0, \frac{1}{n}, \frac2n, \ldots, 1\}^2$ be an $n \times n$ discretization of the unit square. If we minimize over all $(\gamma, \tau) \in \sF$, of the function $\hat{f}(\gamma, \tau)$ defined in Lemma \ref{lem:approx2}, then the minimum value satisfies
$$\min_{(\gamma, \tau) \in \sF} \hat{f}(\gamma, \tau) \leq \min_{(\gamma, \tau) \in [0,1]^2} f(\gamma, \tau) + \frac{2}{n} + \frac{5}{4m}.$$
\end{corollary}
\begin{proof}
By Lemma \ref{lem:approx2}, we have
$\min_{(\gamma, \tau) \in \sF} \hat{f}(\gamma, \tau) \leq \min_{(\gamma, \tau) \in \sF} f(\gamma, \tau) + \frac{5}{4m}.$

Now let $(\gamma^*, \tau^*) = \arg\min_{(\gamma, \tau) \in [0,1]^2} f(\gamma, \tau)$. Let $(\hat{\gamma}, \hat{\tau})$ be the closest point to $(\gamma^*, \tau^*)$ in the discretized grid $\sF$. Then $\abs{\hat{\gamma} - \gamma^*} \leq \frac{1}{2n}$, and $\abs{\hat{\tau} - \tau^*} \leq \frac{1}{2n}$. By Lemma \ref{lem:approx1}, we know $f$ is 1-Lipschitz in $\gamma$ and 3-Lipschitz in $\tau$, which implies that
$$\min_{(\gamma, \tau) \in \sF} f(\gamma, \tau) \leq f(\hat{\gamma}, \hat{\tau}) \leq f(\gamma^*, \tau^*) + \frac{1}{2n} + \frac{3}{2n}.$$
Chaining this with the previous displayed inequality, we obtain
$$\min_{(\gamma, \tau) \in \sF} \hat{f}(\gamma, \tau) \leq f(\gamma^*, \tau^*) + \frac{2}{n} + \frac{5}{4m},$$
as claimed. 
\end{proof}

%% file: computation.tex
In this section, we describe the computations that we performed to obtain a competitive ratio of $0.6629$.\footnote{The code for this paper can be found at \href{https://github.com/MapleOx/rankinglp}{https://github.com/MapleOx/rankinglp}.} Recall that Theorem \ref{thm:main_bound} states that the competitive ratio of RANKING is bounded below by an expression of the form $\min_{(\gamma, \tau) \in [0,1]^2} f(\gamma, \tau)$, where $f$ depends on the function $g$ that the algorithm uses. To obtain our competitive ratio, we
\begin{enumerate}
    \item Solve the LP in Section \ref{sec:lp} for an appropriate discretization of the unit square. (We chose a $50 \times 50$ discretization here.) 
    
    \item Plug the function $g$ obtained from the LP into the bound in Theorem \ref{thm:main_bound}. 
\end{enumerate}

Note that for the function $g$ obtained from the LP, we can only evaluate the bound in Theorem \ref{thm:main_bound} approximately. This is because $g$ is a piecewise-affine function defined by interpolating its values on a $50\times 50$ grid, so it has no amenable closed form. As described in Section \ref{sec:error}, we let $\sF = \{0, \frac1n, \ldots, 1\}^2$ for some large enough $n$, and we evaluate $\min_{(\gamma, \tau) \in \sF} \hat{f}(\gamma, \tau)$, where $\hat{f}$ is an approximation to $f$ amenable to computer evaluation. (Again, refer to Section \ref{sec:error} for the details.)

Corollary \ref{cor:compcheck} gives us quantifiable bound on the error incurred when we evaluate $\min_{(\gamma, \tau) \in \sF} \hat{f}(\gamma, \tau)$ instead of the true bound $\min_{(\gamma, \tau) \in [0,1]^2} f(\gamma, \tau)$. We used a computer to evaluate $\min_{(\gamma, \tau) \in \sF} \hat{f}(\gamma, \tau)$ with $n = 2^{14}$ and $m = 2^{10}$, and obtained $\min_{(\gamma, \tau) \in \sF} \hat{f}(\gamma, \tau) = 0.66433$. Thus, by Corollary \ref{cor:compcheck}, the competitive ratio of the algorithm is at least
$$\min_{(\gamma, \tau) \in \sF} \hat{f}(\gamma, \tau) - \frac2n - \frac{5}{4m} = 0.66298.$$

Computing the bound for the above choice of parameters $n$ and $m$ necessitated the use of clever computation techniques; the naive computation (which simply goes through all $(\gamma, \tau) \in \sF$ one by one, evaluating from scratch $\hat{f}(\gamma, \tau)$ for each) is \emph{too slow} for the size of the discretization we required to obtain a good bound. (For a point $(\gamma, \tau) \in \sF$, we estimate that evaluating $\hat{f}(\gamma, \tau)$ is roughly a $O(m^3)$ operation. The naive computation, which does this for each of the $n^2$ points in $\sF$, is then a $O(n^2m^3)$ computation, which is much too slow for the parameters $n=2^{14}$ and $m = 2^{10}$.) To speed up the computation, we used two techniques:
\begin{enumerate}
    \item Precomputation of values that are used repeatedly by the code, and
    \item Parallelization. 
\end{enumerate}
Even after speeding up the computation using precomputed tables and parallelization, we still needed two days of computing time on a 64-core machine with 64GB of memory.\footnote{We performed this computation on Amazon EC2. We used a compute-optimized \texttt{c6g.16xlarge} instance, running the Amazon Linux 2 AMI.} Without either one of these techniques, the computation would not have terminated in a reasonable amount of time. In the remainder of this section, we describe the above techniques in more detail.  

\begin{remark*}
The perceptive reader might notice that it would be conceptually simpler to skip the second step given above  altogether (i.e.\ plugging the function $g$ from the LP into the bound of Theorem \ref{thm:main_bound}), since the objective of the LP is already an approximation of the bound in Theorem \ref{thm:main_bound}. The reason we do not do this is because to be able to prove a good enough bound, we need to evaluate $\min_{(\gamma, \tau) \in \sF}\hat{f}(\gamma, \tau)$ for a fine enough discretization. However, solving the LP is prohibitively expensive for large discretizations. To put this into context, we solved the LP on a $50 \times 50$ discretization, and used the output of the LP to evaluate $\min_{(\gamma, \tau) \in \sF} \hat{f}(\gamma, \tau)$ on an $n \times n$ discretization, where $n = 2^{14}$. From the description of the LP in Section \ref{sec:lp}, it can be seen that for an $n \times n$ discretization, the LP has roughly $n^3$ variables and $n^4$ constraints. For $n = 2^{14}$, this would have been too large an LP to solve.
\end{remark*}

\subsection{Precomputing Tables}
The computation we are trying to perform is $\min_{(\gamma, \tau) \in \sF} \hat{f}(\gamma, \tau)$. Recall from Section \ref{sec:error} that $\hat{f}$ is defined as
\begin{align*}
    \hat{f}(\gamma, \tau) = \hat{p}(\gamma, \tau) + \frac{1}{m}\sum_{k=0}^{j-1} \frac{\hat{q}(\gamma, \tau, y_k) + \hat{q}(\gamma, \tau, y_{k+1})}{2}
\end{align*}
where $\hat{p}$ and $\hat{q}$ are defined to be
\begin{itemize}
    \item $\hat{p}(\gamma, \tau) = (1- \gamma)(1-\tau) + (1-\tau) \cdot \frac1m \sum_{k=0}^{i-1} g(x_k, \tau)$, and
    \item $\hat{q}(\gamma, \tau, y) = \min_{k \leq i+1} \left\{1 - g(x_k, y) + \frac{1}{m}\sum_{d = 0}^{k-1} g(x_d, y) + \frac{1}{m} \sum_{d=k}^{i-1} g(x_d, y_{j+1}) \right\}$,
\end{itemize}
where in the above expressions, 
\begin{itemize}
    \item $x_k = y_k = \frac{k}{m}$, and
    \item $i$ and $j$ are defined to be the integers such that $x_i \leq \gamma < x_{i+1}$, and $y_j \leq \tau < y_{j+1}$.
\end{itemize} 

The key observation is that for two different points $(\gamma, \tau)$ and $(\gamma', \tau')$, some parts of the computation of $\hat{f}(\gamma, \tau)$ and $\hat{f}(\gamma', \tau')$ are the same. Thus, we can speed up the code by precomputing these values and storing them in memory, so that they can be fetched instead of being recomputed each time they are needed. We identified two types of values that could be reused, and precomputed a table for each. 

\textbf{A table to store the values of g.}
From the expression for $\hat{f}$, we see that it involves many evaluations of $g$. We can precompute these values and store them in an table for future use. Note that we only ever need to evaluate $g$ on points of the form $(x_i, y_j) = (i/n, j/n)$, which results in a $(n+1)\times (n+1)$ table to be stored in memory. For our choice of $n = 2^{14}=16384$, this resulted in a table of size roughly 6GB.

\textbf{A table to store the values of the inner minimum.}
We also precomputed a table to store the values of $\hat{q}(\gamma, \tau, y)$. Note that $\hat{q}(\gamma, \tau, y)$ only depends on $x_i$, $y_j$, and $y$, where $x_i \leq \gamma < x_{i+1}$, and $y_j \leq \tau < y_{j+1}$. That is, the computation of $\hat{q}(\gamma, \tau, y)$ rounds $\gamma$ and $\tau$ to the nearest points on the $\frac1m$ discretized grid. Thus there are $m+1$ possible values for each of $x_i$, $y_j$, and $y$, which implies that there are $(m+1)^3$ possible values for $\hat{q}(\gamma, \tau, y)$. We precomputed a table to store all of these values in memory. For our choice of $m = 2^{10} = 1024$, this resulted in a table of size roughly 8GB. 

\subsection{Parallelization}
The other efficiency gain came from parallelizing the code. With the use of  precomputed tables, our code runs in two stages:
\begin{enumerate}
    \item Stage 1: Compute the two tables described above.
    \item Stage 2: Use the precomputed tables to compute $\min_{(\gamma, \tau) \in\sF} \hat{f}(\gamma, \tau)$. 
\end{enumerate}
Both stages are amenable to parallelization. For the first stage, computing the value of a table entry is independent of computing the value of another table entry, so filling each table can be done in parallel. For the second stage, evaluating $\hat{f}(\gamma, \tau)$ is independent of evaluating $\hat{f}(\gamma', \tau')$ for different pairs $(\gamma, \tau)$ and $(\gamma', \tau')$, so this can also be done in parallel. We used the \texttt{multiprocessing} module in Python to parallelize our code, which we then ran on a 64-core machine on Amazon's EC2. In total, this took about 2 days. Parallelizing was an essential step for making the code run in a reasonable amount of time; extrapolating from its running time on the 64-core machine, we estimate that a non-parallelized version would have taken more than 100 days to run. 



%% file: upper.tex
In this section, we show that our approach cannot obtain a competitive ratio significantly better than 0.6629. Thus, any further progress beyond this bound will require either further weakening in the assumptions of $g$, or a stronger analysis than that of Huang et al. Precisely, we prove the following theorem.

\begin{restatable}{theorem}{upper}
\label{thm:upper}
For any $g$ function that satisfies conditions 1--5, the value of the the bound in Theorem \ref{thm:main_bound} is at most $0.6688$. 
\end{restatable}
\begin{proof}

We can write an LP for which any function $g$ that satisfies conditions 1--5 is feasible, and whose objective upper bounds the value of the bound in Theorem \ref{thm:main_bound} for $g$. This LP is a slight modification of the one in Section \ref{sec:lp}. The reason we must modify the LP slightly is to ensure that
\begin{enumerate}
    \item It is a relaxation of the original problem, in the sense that any function $g$ which satisfies conditions 1--5 must be feasible to the LP, and
    \item The objective of the LP is provably an upper bound on the value of the expression in Theorem \ref{thm:main_bound} when evaluated on $g$. 
\end{enumerate}
Below, we explain how this modification works. 

First, recall that the LP in Section \ref{sec:lp} searched for $g$ by discretizing the unit square into an $n \times n$ grid, for some positive integer $n$, and instantiating variables for the values of $g$ on the discretized grid. Our modified LP also does this. Using the same notation as in Section \ref{sec:lp}, let $x_i = y_i = \frac1n$ for $i = 0, \ldots, n$. Then the variables in our LP will be $g(x_i, y_j)$. 

\textbf{Constraints.} The original LP had 5 sets of constraints, modelling the conditions 1--5. The modified LP takes these 5 sets of constraints, and changes two of them slightly to obtain a valid relaxation. Below are the constraints. Constraints 1, 2, and 5 are unchanged from the LP in Section \ref{sec:lp}, but we will write them again here for completeness. Constraints 3 and 4 are the ones that have been changed slightly. As we describe each constraint, we explain why it is a valid relaxation of the corresponding condition.
\begin{enumerate}
    \item \underline{$g(x, y): [0,1]^2 \to [0,1]$ and $g$ is continuous.} The corresponding LP constraints are $0 \leq g(x_i, y_j) \leq 1$, for all $i,j = 0, 1, \ldots, n$. This constraint is clearly a valid relaxation. 
    \item \underline{$g(x, y)$ is increasing in $x$ and decreasing in $y$.} The corresponding LP constraints are 
    \begin{itemize}
        \item $g(x_i, y_j) \leq g(x_k, y_j)$ for all $0 \leq i, j, k \leq n$ with $i \leq k$;
        \item $g(x_i, y_j) \geq g(x_i, y_l)$ for all $0 \leq i, j, l \leq n$ with $j \leq l$.
    \end{itemize} 
    This constraint is a valid relaxation, because it is checking monotonicity at only a finite set of points. 
    \item \underline{$\frac{\partial g(x, y)}{\partial x} \leq g(x, y)$}. We discretize this constraint to create the following LP constraint:
    $$\frac{g(x_{i+1}, y_j) - g(x_i, y_j)}{x_{i+1} - x_i} \leq g(x_{i+1}, y_j) \quad \text{for all $0 \leq i \leq n-1$, $0 \leq j \leq n$}$$
    To see why this constraint is a valid relaxation, we use the mean value theorem. The means value theorem tells us that
    $$
    \frac{g(x_{i+1}, y_j) - g(x_i, y_j)}{x_{i+1} - x_i} = \frac{\partial g(x, y_j)}{\partial x}
    $$
    for some $x \in [x_i, x_{i+1}]$. If $g$ satisfies conditions 1--5, then $\frac{\partial g(x, y_j)}{\partial x} \leq g(x, y_j) \leq g(x_{i+1}, y_j)$, where the first inequality is by condition 3 and the second inequality is by condition 2. Thus the discretized constraint is indeed a valid relaxation. 
    
    \begin{remark*}
    The LP in Section \ref{sec:lp} had instead the right-hand side of the constraint as $g(x_i, y_{j+1})$. Since $g(x_i, y_{j+1}) \geq g(x_{i+1}, y_j)$, the constraints of the modified LP are weaker. We do this to make the modified LP a valid relaxation. 
    \end{remark*}

    \item \underline{$\frac{\partial g(x, y)}{\partial y} \geq g(x, y) - 1$}. Similar to the previous constraint, the corresponding LP constraint is

    $$\frac{g(x_i, y_{j+1})-g(x_i, y_j)}{y_{j+1}-y_j} \geq g(x_{i}, y_{j+1}) - 1 \quad \text{for all $0 \leq i, j \leq n-1$}$$
    An argument similar to the one for the previous constraint shows that this one is a valid relaxation. 
    \item \underline{For all $x, y, y'$ with $y' > y$,
    $g(1, y) - g(x, y) \geq g(1, y') - g(x, y')$}. The corresponding LP constraints are
    $$g(x_n, y_j) - g(x_i, y_j) \geq g(x_n, y_l) - g(x_i, y_l) \quad \text{for all $0 \leq i,j,k \leq$ with $l > j$.}$$
    This is a valid relaxation, because it is checking condition 5 at only a finite set of points. 
\end{enumerate}
Now that we have described the constraints of the modified LP, and proved that any feasible function $g$ must satisfy them, we turn to describing the objective of the modified LP. 

\textbf{Objective}.
Recall that the LP in Section \ref{sec:lp_objective} approximated the objective in Theorem \ref{thm:main_bound} by 
\begin{enumerate}
    \item Replacing the $\min_{0\leq \gamma, \tau \leq 1}$ and $\min_{\theta \leq \gamma}$ expressions with minimums over a finite set of values, and 
    \item Replacing the integrals with finite sums. 
\end{enumerate}
For the modified LP, we also replace the min expressions with minimums over a finite set of values as before. Since this can only increase the objective, this is valid for obtaining an upper bound. To get an upper bound on the integrals, we replace them with right Riemann sums whenever the integrand is monotone increasing. It turns out that all but one of the integrals in the objective has a monotone increasing integrand. For the remaining integral (which is not necessarily monotone increasing or monotone decreasing), we replace it with a trapezoidal sum and bound the error using a Lipschitz continuity argument. 

More precisely, recall that the bound in Theorem \ref{thm:main_bound} is
$$\min_{0 \leq \gamma, \tau \leq 1} f(\gamma, \tau),$$
where
\begin{align*}
&f(\gamma, \tau) = (1-\tau)(1-\gamma) + (1-\tau)\int_0^\gamma g(x, \tau) dx \\ &+ 
\int_0^\tau \min_{\theta \leq \gamma} \left\{(1-g(\theta, y)) + \int_0^\theta g(x, y)dx + \int_\theta^\gamma g(x, \tau) dx\right\} dy 
\end{align*}
Below, we describe how we express the objective in the modified LP. We justify that each step can only increase the objective. 
\begin{itemize}
    \item Replace $\min_{0 \leq \gamma,\tau \leq 1} f(\gamma, \tau) $ by $\min_{(i,j) \in S} f(x_i, y_j)$, for some small finite subset $S \subset \{0, \ldots, n\}^2$. This can only increase the objective. The way we chose the subset $S$ is by observing, empirically, where the expression $\min_{0 \leq \gamma, \tau \leq 1} f(\gamma, \tau)$ attains its minimum. (We did this by  computing $\min_{0 \leq i,j \leq n} f(x_i, y_j)$ for several small values of $n$, and seeing at which points $(x_i, y_j)$ in the unit square the minimum was attained.) We then chose several of these points to put into $S$; in the end our set $S$ contained 7 points.\footnote{These points were $(0, 0), (1, 1), (0, 1), (1, 0), (\frac{23}{40}, \frac{27}{40}), (\frac12, \frac34), (\frac{13}{30}, \frac{23}{30})$.} The reason why we opted to make the minimum over such a small set $S$ (instead of a discretization of the entire unit square) is because this led to a substantially smaller LP while only increasing the objective value by a negligible amount. The full LP is very computationally expensive to solve, so this type of heuristic was necessary for obtaining an upper bound in a reasonable amount of computation time.  
    \item Replace $\min_{\theta \leq \gamma}$ by $\min_{\theta \in \{0, \gamma\}}$. Again, this can only increase the objective. We opted to take a minimum over just the two values $\{0, \gamma\}$, instead of a discretization of the interval $[0, \gamma]$, because this led to a substantially smaller LP while increasing the objective value by only a negligible amount. Empirically, we noticed that for the function $g$ returned by the LP, the $\min_{\theta \leq \gamma}$ was usually attained at $\theta = 0$ or $\theta = \gamma$.
    \item Replace the integral 
    $$\int_0^\tau \min_{\theta \in\{0, \gamma\}} \left\{(1-g(\theta, y)) + \int_0^\theta g(x, y)dx + \int_\theta^\gamma g(x, \tau) dx\right\} dy$$ with a trapezoidal sum. That is, if $h(y)$ is the integrand, then we replace the integral $\int_0^\tau h(y) dy$ with $\frac{1}{n} \sum_{j=0}^{l-1} \frac{h(y_{j+1}) + h(y_j))}{2}$, where $y_l = \tau$. This by itself may not result in an upper bound -- in particular, it is not clear if $h(y)$ is monotone in $y$. However, it cannot be too far from an upper bound, and we can quantify this using a Lipschitz argument.
    
    Indeed, note that $h(y)$ is 1-Lipschitz in $y$. (The argument for this is the same as one used in the proof of Lemma \ref{lem:approx2}, so we choose to not repeat it here.) Hence, by Lemma \ref{lem:lipschitz_integral} in Appendix \ref{app:lipschitz}, we have $\int_{y_j}^{y_{j+1}} h(y) dy \leq \frac{h(y_{j+1}) + h(y_j)}{2n} + \frac{1}{4n^2}$. (Equality is attained when $h(y_j) = h(y_{j+1})$, and $h$ is linear with slope 1 on the interval $[y_j, y_j + \frac{1}{2n}]$, and linear with slope $-1$ on the interval $[y_j + \frac{1}{2n}, y_{j+1}]$.) 
    
    Therefore, 
    $$\int_0^{y_l} h(y) dy \leq \frac{1}{n} \sum_{j=0}^{l-1} \frac{h(y_{j+1}) + h(y_j)}{2} + \frac{l}{4n^2} \leq \frac{1}{n}\sum_{j=0}^{l-1} \frac{h(y_{j+1}) + h(y_j)}{2} + \frac{1}{4n}.$$
    Thus, we can replace the integral $\int_0^{y_l} h(y) dy$ with the sum $\frac{1}{n}\sum_{j=0}^{l-1} \frac{h(y_{j+1}) + h(y_j)}{2} + \frac{1}{4n}$, and this will be a valid upper bound.

    \item Replace the three other integrals (which are over $x$) with right Riemann sums over the discretized grid. Since those three integrals all have a monotone increasing integrand (as $g(x, y)$ is increasing in $x$), it follows that replacing them with right sums can only increase the objective. For example, the integral $\int_0^{x_k} g(x, \tau) dx$ is replaced by the sum $\frac1n \sum_{i=1}^k g(x_i, \tau)$. 
    
\end{itemize}

\textbf{Putting together.} Putting everything together, the upper bound LP is as follows:
\begin{align*}
    \max \quad &t \\
    \text{s.t.} \quad &t \leq (1-x_k)(1-y_l) + (1-y_l)\cdot \frac1n \sum_{i=1}^k g(x_i, y_l)  + \frac{1}{n} \left(\sum_{j=0}^{l-1} \frac{\tilde{h}(y_{j+1}) + \tilde{h}(y_j)}{2}\right) + \frac{1}{4n}\quad \text{$\forall k,l \in [n]$} \\
    &\text{and $g$ satisfies the constraints described above.}
\end{align*}
In the above, $\tilde{h}$ is the function $h$ with the integrals replaced by right Riemann sums:
$$\tilde{h}(y_j) := \min_{x_i \in \{0, x_k\}} \left\{1 - g(x_i, y_j) +\frac1n \sum_{d=1}^i g(x_d, y_j) + \frac1n \sum_{d=i+1}^k g(x_d, y_l) \right \}$$
We solved the LP with $n = 210$, and obtained an objective value of 0.6688. This shows that \emph{any} function $g$ satisfying conditions 1--5 has value at most 0.6688 when plugged into the bound in Theorem \ref{thm:main_bound}. 

\end{proof}

%% file: conclusion.tex
Figure \ref{fig:comparison} compares the contour plot of the function $g$ we used, to the function $g(x, y) = \frac12(h(x) + 1 - h(y))$ used by Huang et al.\ (Here, $h(x) = \min(1, \frac12 e^x).)$ The plots look qualitatively quite different. One interesting question is to try to extrapolate a simple function $g$ from the LP contour plot, such that when $g$ is plugged into the bound in Theorem \ref{thm:main_bound}, it improves upon the competitive ratio in Huang et al.\ Visually, the LP contour plot suggests trying a piecewise-linear $g$ with two pieces, where one piece only depends on $y$. However, we can prove that no function in this class can improve upon the competitive ratio in Huang et al.\

As we have noted, because of our upper bound, the value of the competitive ratio cannot be improved by much without either weakening the assumptions on $g$ that we use in Theorem \ref{thm:main_bound} or improving the analysis of Huang et al.\ \cite{HTWZ19} that we use.  Huang et al.\ give a potentially stronger bound on the competitive ratio in the conclusion of their paper.  However, it was unclear to us how to express their stronger bound as a linear program on discretized values of $g$.

As discussed in Section \ref{sec:upper}, in order to compute the upper bound, we significantly relaxed the points at which the minimums of the function $f(\gamma, \tau)$ are taken, yet this did not change the value of the LP by much.  It seems possible that the Huang et al.\ analysis can be simplified to reflect this fact.

It would also be interesting to derive an improved upper bound for this problem.  To our knowledge, the best upper bound is the same as for the unweighted online bipartite matching problem with random arrivals, which is 0.823, and is due to Manshadi, Oveis Gharan, and Saberi \cite{MOS12}.

\begin{figure}[t!]
    \centering
    \begin{subfigure}[t]{0.45\textwidth}
        \centering
        \includegraphics[height=2in]{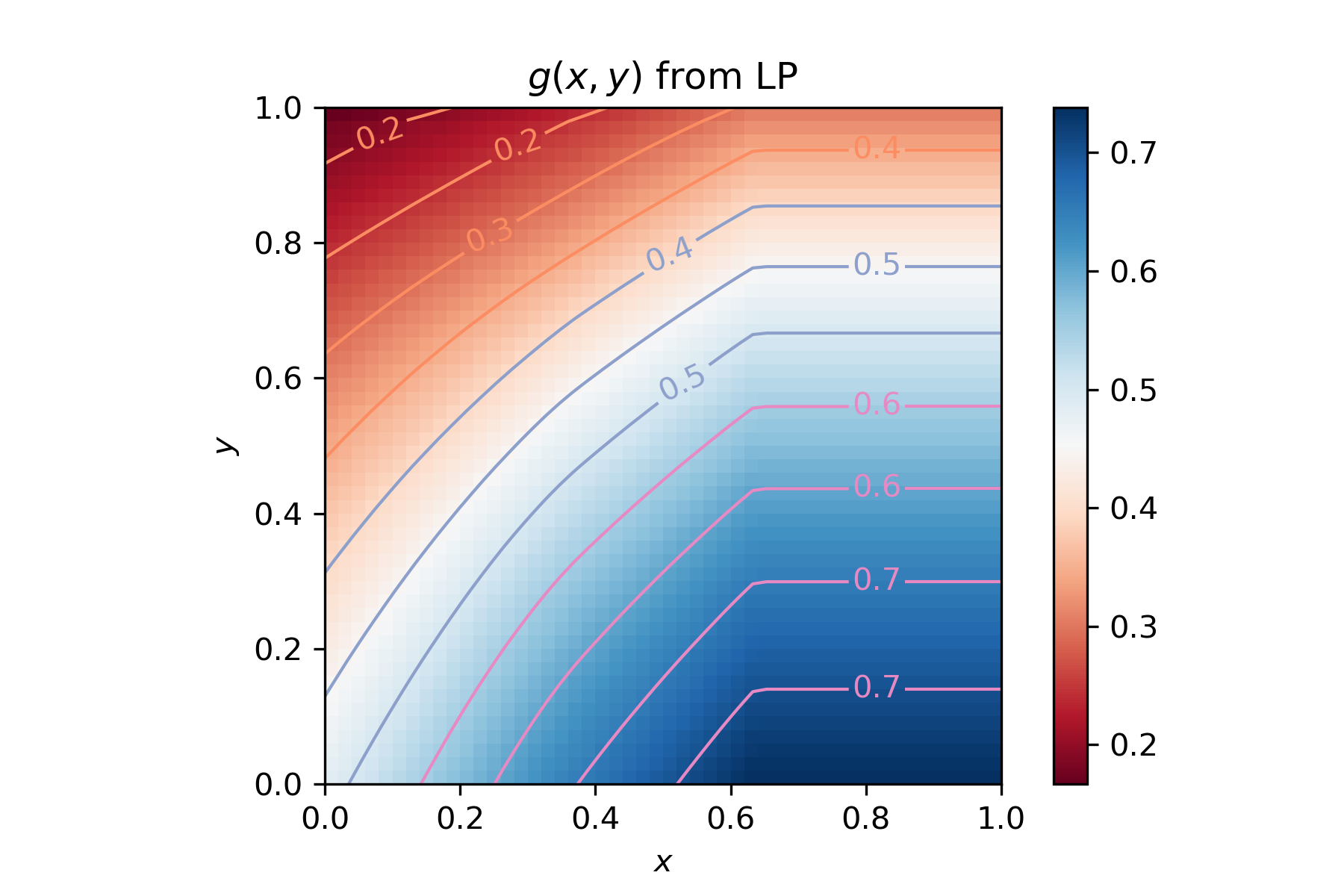}
        \caption{Contour plot of the function values obtained from solving the LP in Section \ref{sec:lp}. Here, we used a $50 \times 50$ discretization.}
    \end{subfigure}%
    ~ 
    \begin{subfigure}[t]{0.45\textwidth}
        \centering
        \includegraphics[height=2in]{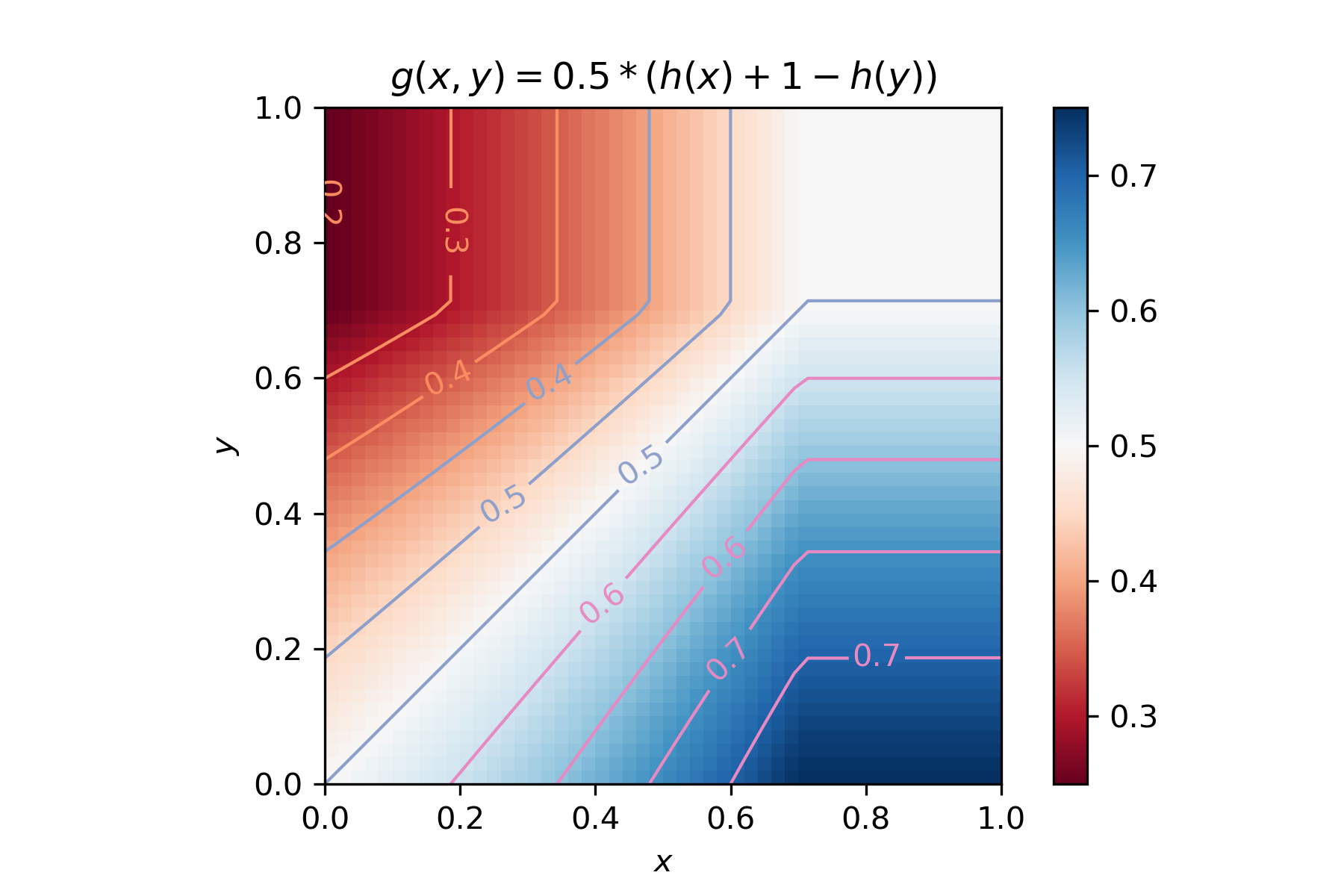}
        \caption{Contour plot of the function $g(x, y) = \frac12(h(x) + 1 - h(y))$, where $h(x) = \min(1, \frac12 e^x)$. This is the function used by Huang et al.}
    \end{subfigure}
    \caption{Side-by-side comparison of the function $g$ we used, versus the function $g$ used by Huang et al.}
    \label{fig:comparison}
\end{figure}

%% file: lipschitz.tex
In this section, we show that integrals of Lipschitz functions are well-approximated by trapezoidal sums. The below lemma is used twice in our previous proofs, once in the proof of Lemma \ref{lem:approx2}, and once in the proof of Theorem \ref{thm:upper}.
\begin{lemma}
\label{lem:lipschitz_integral}
Suppose $f: \R \to \R$ is $L$-Lipschitz. 
Let $n$ be a positive integer, and let $x_k = \frac{k}{n}$ for all integers $k$. Then for all $i < j$, we have
$$
\abs{\int_{x_i}^{x_j} f(x) dx - \frac{1}{2n}\sum_{k=i}^{j-1}\left(f(x_k) + f(x_{k+1})\right)} \leq \frac{L(j-i)}{4n^2}.
$$
\end{lemma}
\begin{proof}
It suffices to show that
$$\abs{\int_{x_k}^{x_{k+1}} f(x) dx - \frac{f(x_k) + f(x_{k+1})}{2n}} \leq \frac{L}{4n^2}.$$
The statement in the lemma then follows from summing over $k$, and then applying the triangle inequality. 

For clarity of exposition, we will prove the following simpler claim. The proof of this claim generalizes easily to the more general statement above. 

\begin{claim}
Whenever $f$ is 1-Lipschitz,
$$\int_0^1 f(x) dx - \frac{f(0) + f(1)}{2} \leq   \frac{1}{4}$$
\end{claim}
\begin{proof}
Define $g_1(x)$ to be the linear function through $(0, f(0))$ with slope 1. Similarly, define $g_2(x)$ to be the linear function passing through $(1, f(1))$ with slope $-1$. Since $f$ is 1-Lipschitz, it follows that $f(x) \leq \min\{g_1(x), g_2(x)\}$ on the interval $[0, 1]$. It follows that 
$\int_0^1 f(x) dx \leq \int_0^1 \min\{g_1(x), g_2(x)\} dx.$ (For an illustration, see Figure \ref{fig:lip_int}.)

\begin{figure}[H]
\begin{center}
\begin{tikzpicture}
\vertex (a) at (-3, -3) [label = below left:{$(0, f(0))$}] {};
\vertex (b) at (4, -4) [label=below right:{$(1, f(1))$}] {};
\vertex (c) at (0, 0) {};
\draw[domain=-3:1, smooth, variable=\x, thick, blue] plot ({\x}, {\x});
 \draw[domain=-1:4, smooth, variable=\x, thick, blue]  plot ({\x}, {-\x});
 \draw[domain=-3:4, smooth, variable=\x, thick]  plot ({\x}, {-\x/7-3+-3/7});
 
 \node at (-1.75,-1) {\footnotesize $g_1(x)$};
 \node at (2.25, -1.5) {\footnotesize $g_2(x)$};
  \node at (0, -4) {\footnotesize $f(x)$};
  
\end{tikzpicture}
\caption{The 1-Lipschitz function $f$, together with its two upper bounds $g_1$ and $g_2$.}
\label{fig:lip_int}
\end{center}

\end{figure}
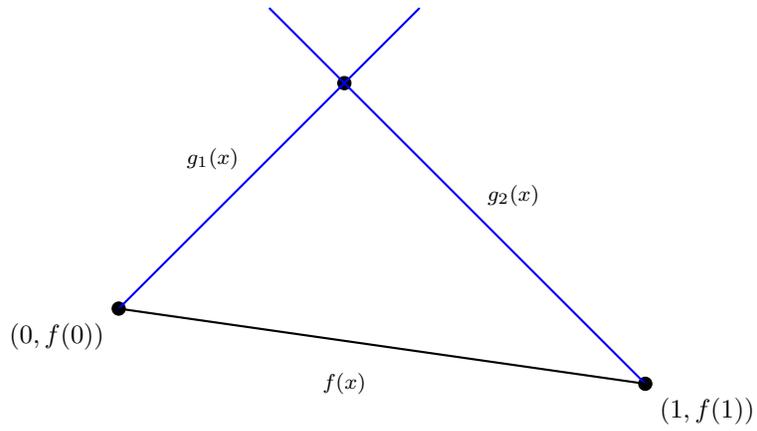

It is a straightforward exercise to calculate that
$$\int_0^1 \min\{g_1(x), g_2(x)\} dx - \frac{f(0) + f(1)}{2} = \frac14\cdot \left(1 - (f(1)-f(0))^2\right) \leq \frac14,$$
which completes the proof of the claim.
\end{proof}

\end{proof}